\documentclass[11pt,a4paper,notitlepage]{article}

\usepackage[a4paper,margin=2.5cm]{geometry}
\usepackage[utf8]{inputenc}
\usepackage{listings}
\usepackage{orcidlink}

\usepackage{xspace}
\usepackage{amssymb}
\usepackage{amsmath}
\usepackage{amsthm}

\usepackage{enumerate}

\usepackage[ruled]{algorithm2e}

\usepackage{tikz}
\usepackage{subcaption}
\usepackage{listings}

\newcommand{\nc}[1]{\newcommand{#1}}
\nc{\Z}{\ensuremath{\mathbb{Z}}}
\nc{\nat}{\ensuremath{\mathbb{N}}}
\nc{\natpos}{\ensuremath{\nat_{\geqslant 1}}}
\nc{\finsubseteq}{\ensuremath{\subseteq_{\text{fin}}}}
\nc{\finsubset}{\ensuremath{\subset_{\text{fin}}}}
\nc{\contResp}[1]{\ensuremath{\subseteq_{#1}}}
\nc{\ContResp}[1]{\ensuremath{\supseteq_{#1}}}
\nc{\propContResp}[1]{\ensuremath{\subset_{#1}}}

\nc{\dc}{\ensuremath{\chi}} 

\nc{\set}[1]{\ensuremath{\{#1\}}}
\nc{\inftyset}{\ensuremath{\set{\infty}}}
\nc{\setc}[2]{\set{#1 \ : \ #2}}
\nc{\deff}{\ensuremath{:=}}
\let\P\relax
\DeclareMathOperator{\P}{\mathcal{P}}
\DeclareMathOperator{\Pfin}{\mathcal{P}_{\text{fin}}}
\DeclareMathOperator{\Pfinplus}{\mathcal{P}^{+}_{\text{fin}}}

\nc{\soi}{soi}

\nc{\swgqueries}{swg-queries\xspace}
\nc{\swgquery}{swg-query\xspace}

\nc{\swggqueries}{swgg-queries\xspace}
\nc{\swggquery}{swgg-query\xspace}

\nc{\dswgqueries}{dswg-queries\xspace}
\nc{\dswgquery}{dswg-query\xspace}

\nc{\query}{\ensuremath{q}}
\nc{\ws}{\ensuremath{w}} 
\nc{\wc}{\ensuremath{c}} 
\nc{\wC}{\ensuremath{C}} 
\nc{\WC}{\ensuremath{\tilde{c}}} 
\nc{\types}{\textit{types}} 
\nc{\typesets}{\textit{typesets}} 
\nc{\vars}{\textit{vars}} 

\nc{\TS}{\ensuremath{\Gamma}} 
\nc{\TSalt}{\ensuremath{\Delta}} 
\nc{\VS}{\ensuremath{\textsf{\upshape Vars}}} 
\nc{\Vars}{\VS} 
\nc{\type}{\ensuremath{\gamma}}
\DeclareMathOperator{\ta}{\mathsf{a}}
\DeclareMathOperator{\tb}{\mathsf{b}}
\DeclareMathOperator{\tc}{\mathsf{c}}
\DeclareMathOperator{\td}{\mathsf{d}}

\nc{\tSUB}{\mathsf{SUB}}
\nc{\tSCH}{\mathsf{SCH}}
\nc{\tCHE}{\mathsf{CHE}}
\nc{\tRES}{\mathsf{RES}}
\nc{\tSUS}{\mathsf{SUS}}
\nc{\tEVI}{\mathsf{EVI}}
\nc{\tFAI}{\mathsf{FAI}}
\nc{\tFIN}{\mathsf{FIN}}
\nc{\tKIL}{\mathsf{KIL}}
\nc{\tLOS}{\mathsf{LOS}}
\nc{\tUPD}{\mathsf{UPD}}
\nc{\nS}{\mathsf{nTypes}}
\nc{\base}{\mathsf{base}}

\DeclareMathOperator{\subseq}{\preccurlyeq}

\nc{\Mods}[2]{\ensuremath{\mathsf{Mod}_{#1}(#2)}}
\nc{\mods}[1]{\Mods{\TS}{#1}}
\nc{\Mintrace}[2]{\ensuremath{\mathsf{Min}_{#1}(#2)}}
\nc{\mintrace}[1]{\Mintrace{\TS}{#1}}
\nc{\Hom}{\ensuremath{\stackrel{\textup{hom}}{\longrightarrow}}}

\nc{\trace}{\ensuremath{t}}
\nc{\Sample}{\ensuremath{\mathcal{S}}}
\nc{\supp}{\mathsf{sp}} 
\nc{\Supp}{\mathsf{supp}} 

\nc{\matchProb}{\mathsf{Match}}
\nc{\compDescProb}{\mathsf{CompDescQuery}}
\nc{\checkDescProb}{\mathsf{CheckDescQuery}}

\DeclareMathOperator{\npclass}{\mathsf{NP}}

\nc{\Pos}[2]{\ensuremath{\textit{pos}{(#1,#2)}}}
\nc{\replace}[3]{\ensuremath{#1\langle #2 \mapsto #3\rangle}}
\nc{\mgq}{\ensuremath{q_{\textit{mg}}}} 
\nc{\mgs}{\ensuremath{s_{\textit{mg}}}} 
\nc{\NR}{\ensuremath{U}} 
\nc{\AV}{\ensuremath{V}} 
\nc{\Error}{\ensuremath{\bot}}
\nc{\QComputeDescrQueryFromq}{\textsf{\textup{DescrQuery($\Sample$,$\supp$,$(\ell,\ws,\WC,k)$)}}}
\nc{\QComputeDescrQueryFromqShort}{\textsf{\textup{DescrQuery}}}
\nc{\TypeRepl}{\textsf{TypeRep}}
\nc{\VarRepl}{\textsf{VarRep}}
\nc{\NoChange}{\textsf{NoChange}}

\nc{\dis}{\mathsf{dis}}
\nc{\gen}{\mathsf{gen}}
\nc{\cand}{\mathsf{cand}}
\nc{\TSaltCalc}{\textsf{\textup{$\TSalt$-Calculation($\Sample$,$\supp$,$\dis$)}}}
\nc{\TSaltCalcShort}{\textsf{\textup{$\TSalt$-Calculation}}}

\newtheorem{theorem}{Theorem}{\bfseries}{\itshape}
\newtheorem{Lemma}[theorem]{Lemma}{\bfseries}{\itshape}
\newtheorem{Definition}[theorem]{Definition}{\bfseries}{\itshape}
\newtheorem{Example}[theorem]{Example}{\bfseries}{\itshape}
{\bfseries}{\itshape}
\newtheorem{Proposition}[theorem]{Proposition}{\bfseries}{\itshape}
\newtheorem{Claim}[theorem]{Claim}{\bfseries}{\itshape}
\newtheorem{Corollary}[theorem]{Corollary}{\bfseries}{\itshape}
\newtheorem{Remark}[theorem]{Remark}{\bfseries}{\itshape}

\title{Puzzling over Subsequence-Query Extensions: Disjunction and Generalised Gaps\footnote{This article is the full version of a contribution accepted at the 15th Alberto Mendelzon International Workshop on Foundations of Data Management (2023). Both authors contributed equally.}}

\begin{document}

\author{André Frochaux\,\orcidlink{0009-0001-7918-8725}\ and Sarah Kleest-Meißner\,\orcidlink{0000-0002-4133-7975} \\Humboldt-Universit\"at zu Berlin, Germany,\\ \texttt{\{andre.frochaux|kleemeis\}@informatik.hu-berlin.de}}

\date{}
\maketitle

\begin{abstract}
  \noindent A query model for sequence data was introduced in \cite{KSSSW22} 
    in the form of subsequence-queries with wildcards and gap-size 
    constraints (\swgqueries, for short).
    These queries consist of a pattern over an alphabet of variables and types, 
    as well as a global window size and a number of local gap-size constraints.
    We propose two new extensions of \swgqueries, which both enrich the 
    expressive power of \swgqueries{} in different ways: 
    subsequence-queries with generalised gap-size constraints 
    (\swggqueries, for short) and disjunctive subsequence-queries 
    (\dswgqueries, for short). 
    We discuss a suitable characterisation of containment, a classical property
    considered in database theory, and adapt results concerning
    the discovery of \swgqueries{} to both, \swggqueries{} and \dswgqueries.
\end{abstract}

\section{Introduction}

Applications in different domains like 
cluster monitoring \cite{VPKOTW2015}, 
urban transportation \cite{ArtikisWSBLPBMKMGMGK14}, and 
in finance\cite{TeymourianRP12}, 
use models for sequence data, which define an order for a set of data items 
\cite{BabcockBDMW02}.
Respective systems enable the definition of queries which detect patterns of
data items describing a \emph{situation of interest} (\emph{\soi} for short), 
for example error occurence, in a specific order and temporal context. 
\par 

Finding a suitable query is a non-trivial task. A user may know the time at 
which a certain job fails execution, but does not exactly conceive a situation
which forcasts the failure. 
It was therefore suggested to automatically discover a query from historic 
sequence data which describes the \soi. Such a query may then be used in 
pro-active applications where they shall anticipate a \soi{} to prepare for 
it accordingly \cite{ArtikisBBWEFGHLPSS14}.
\par 

In \cite{KSSSW22} a formal model (referred to as \emph{\swgqueries}) was 
proposed, which covers the essence of discovering a query from sequence data. 
In a nutshell, an \swgquery{} consists of a pattern over an alphabet of 
variables and types, a global window size and a tuple of local gap-size 
constraints.
Syntactically, \swgqueries{} are so-called Angluin-style patterns with 
variables, but with a semantics adapted to sequence data: each variable in the 
query string ranges only over a single symbol and the query matches if, after 
replacing the variables by single data items, it occurs as a subsequence that 
satisfies the window size and local gap-size constraints.
Angluin-style patterns were introduced in \cite{Angluin80} and play a central 
role for inductive inference, in formal language theory and combinatorics on 
words (see \cite{ShinoharaArikawa1995}, \cite{ManeaSchmid2019}, 
\cite{MateescuSalomaa1997}).
%
Concepts and algorithms from inductive inference of so-called pattern languages,
that can be described by Angluin-style patterns, can be adapted to \swgqueries.
Especially the notion of \emph{descriptive patterns} (already introduced in 
\cite{Angluin80}, see also \cite{FreydenbergerReidenbach2010}, 
\cite{FreydenbergerReidenbach2013}) forms a key concept and enables the 
adaptation of \emph{Shinohara's algorithm} \cite{Shinohara82} for Angluin-style 
patterns. This algorithm computes a descriptive Angluin-style pattern upon input
of a  
finite set of sequences of data items.
The corresponding adaptation to \swgqueries{} including some extensions were 
presented in \cite{KSSSW22}, and liftetd to a multi-dimensional data model in 
\cite{KSSSW23}.
\par 

Subsequences in general have extensively been studied both in a 
purely combinatorical sense (in formal language theory, logic and combinatorics 
on words) and algorithmically (in string algorithms and bioinformatics); 
see the introductions of the recent papers \cite{GawrychowskiEtAl2021}, 
\cite{DayEtAl2021} for a comprehensive list of relevant pointers. 
The problem of matching subsequences with gap-constraints (and analysis problems
with respect to the set of all gap-constrained subsequences of given strings) 
has been investigated in the recent papers \cite{DayEtAl2022}, 
\cite{KoscheEtAl2022} (see also \cite{KoscheEtAl2022b} for a survey). 
\par 

Queries defined for complex event recognition (\emph{CER}, for short) usually 
use operators such as sequencing, conjunction and disjunction, Kleene closure, 
negation and variables which may be bound to data items in a stream 
\cite{GiatrakosAADG20}.
Inspired by the generalised gap-size constraints described in \cite{DayEtAl2022}
that are defined over strings other than patterns over variables and types, 
and the use of disjunction in CER languages, we introduce two new notions of 
subsequence-queries, which both extend the expressive power of \swgqueries{} in 
different ways:
\begin{itemize}
    \item \emph{disjunctive subsequence-queries with wildcards and local 
        gap-size constraints}, for short: \emph{\dswgqueries}, and
    \item \emph{subsequence-queries with wildcards and generalised gap-size 
        constraints}, for short: \emph{\swggqueries}
\end{itemize}

\noindent Improving the expressive power of the underlying language used for an 
automatically discovered decriptive query leads to results of increased 
precision. Enabling disjunction is a natural and effective way to reach this. 
Our second approach of generalised gap-size constraints allows detecting
temporal contexts not only between consecutive data items, but between any 
data items in the query string.

The remainder of this paper is structured as follows. 
Section~\ref{sec:traces-and-queries} introduces both, \swggqueries{} and 
\dswgqueries, and discusses the relation to \swgqueries. 
In Section~\ref{sec:discovery-short}, we provide a solution for the query 
discovery problem for both kinds of queries.
Section~\ref{sec:conclusion} concludes the paper. 
Due to space limitations of the conference version, proof details had to be deferred to the appendix.

\section{Traces and Queries}\label{sec:traces-and-queries}

This section introduces the syntax and semantics of both, 
\swggquery (Section \ref{subsec:swggqueries}), and
\dswgqueries (Section \ref{subsec:dswgqueries}).
For a better understanding we consider each extension individually.

By $\Z$, $\nat$, $\natpos$ we denote the set of integers, non-negative
integers, and positive integers, respectively. For every set $M$ we denote the powerset by $\P(M)$, i.e. the set of all subsets of $M$, and $\Pfin(M) \deff \setc{X \in \P(M)}{X \text{ is finite}}$ is the set of all finite subsets from $M$. Moreover, we write $\Pfinplus(M)$ for $(\Pfin(M)\setminus \varnothing)$ and for every $k \in \natpos$, we define $\P_k(M)\deff \set{ m \in \Pfin(M)\, | \, |m|=k}$. For $\ell\in\nat$ we
let $[\ell]=\setc{i\in\natpos}{1\leq i\leq \ell}$. 
For a non-empty set $A$ we write $A^*$ (and $A^+$) for the set of all (non-empty) strings built from symbols in $A$.
By $|s|$ we denote the length of a string $s$, and for a position
$i\in[|s|]$ we write $s[i]$ to denote the letter at position $i$ in $s$.
A \emph{factor} of a string $s\in A^*$ is a string $t\in A^*$ such that $s$ is of the form $s_1 t s_2$ for $s_1,s_2\in A^*$.
\par
An \emph{embedding} is a mapping $e : [\ell] \to [n]$ with $\ell \leq n$ such that $i < j$ implies $e(i) < e(j)$ for all $i, j \in [\ell]$. Let $s$ and $t$ be two strings with $|s| \leq |t|$. We say that $s$ is a \emph{subsequence of $t$ with embedding} $e : [|s|] \rightarrow [|t|]$, if $e$ is an embedding and $s[i] = t[e(i)]$ for every $i \in [|s|]$. We write $s \subseq_e t$ 
to indicate that $s$ is a subsequence of $t$ with embedding $e$; and we write $s \subseq t$ to indicate that there exists an embedding $e$ such that $s \subseq_e t$.\par

We model traces as finite, non-empty strings over some (finite or
infinite) alphabet $\TS$ of \emph{types}.
It will be reasonable to assume that $|\TS|\geq 2$.
A \emph{trace} (over $\TS$) is a string $\trace\in\TS^+$. We write
$\types(\trace)$ for the set of types that occur in $\trace$.
Finally, we fix a countably infinite set $\VS$ of \emph{variables}, and we will always
assume that $\VS$ is disjoint with the set $\TS$ of considered types.

\subsection{Syntax and semantics of \swggqueries}\label{subsec:swggqueries}
\begin{Definition}\label{def:swggquery}
    An \swggquery $\query=(s,\ws,\wC)$ (over $\VS$ and $\Gamma$) is specified by
    \vspace*{-0.5em}
    \begin{itemize}
	\item a query string $s \in (\VS \cup \Gamma)^+$, 
	\item a global window size $\ws \in \natpos \cup \inftyset$ with 
        $w \geqslant |s|$ and
	\item a finite set $\wC$ of generalised gap-size constraints 
        (for $|s|$ and $\ws$) of form 
        \vspace*{-0.25em}
        \[
            (c^-,c^+,r)_{j} \ \in \ (\nat \times \nat \cup \set{\infty} \times \natpos)_{\natpos} 
        \]
        \vspace*{-0.25em}
        for $j \in [|s|-1]$, $r \leq c^- \leq c^+$ and $j + r \leq |s|$. 
    \end{itemize} 
    \vspace*{-0.70em}
\end{Definition}

\noindent The semantics of \swggqueries is defined as follows: each variable in $s$ 
represents an arbitrary type from $\TS$. A query $\query=(s,\ws,\wC)$ 
\emph{matches in a trace $\trace$} (in symbols: $\trace\models\query$), if the 
wildcards in s can be replaced by types in $\TS$ in such a way that the resulting
string $s'$ satifies the fowllowing: 
$\trace$ contains a factor $\trace'$ of length at most $\ws$ such that $s'$ 
occurs as a subsequence in $\trace'$ and for each $(c^-,c^+,r)_j\in\wC$ the gap  
between $s[j]$ and $s[j+r]$ in $\trace'$ has length at least $c^-$ and at most 
$c^+$. Gaps of range $r$ for $r\in [w]$ without any constraints are implicitly 
set to the most general constraint $(0,\infty,r)$.
\par 

A more formal description of these semantics relies on the following additional
notation: An embedding $e: [\ell] \to [n]$ \emph{satisfies a global window size}
$\ws$, if $e(\ell) - e(1) +1 \leq \ws$.
Furthermore, it \emph{satisfies a set of generalised gap-size constraints} 
$\wC$ (for $|\ell|$ and $\ws$) if
$c^- \leq e(j+r)-1-e(j) \leq c^+$, for each $(c^-,c^+,r)_{j}\in\wC$. 
\par

A \emph{substitution} is a mapping $\mu:(\VS\cup\TS) \to (\VS\cup\TS)$ with 
$\mu(\type)=\type$ for all $\type\in\TS$. We lift substitutions to mappings
$(\VS\cup\TS)^+ \to (\VS\cup\TS)^+$ in the obvious way, i.e. 
    $\mu(s) = \mu(s[1]) \mu(s[2]) \ldots \mu(s[\ell])$
for $s\in(\VS\cup\TS)^+$ and $\ell\deff|s|$.
\par

An \swggquery $\query=(s,\ws,\wC)$ \emph{matches in a trace $\trace\in\TS^+$} 
(or $\trace$ \emph{matches} $\query$), if and only if there are a substitution
$\mu:(\VS\cup\TS) \to \TS$ and an embedding $e: [\ell] \to [n]$ that satisfies
$\ws$ and $\wC$, such that $\mu(s)\subseq_e \trace$. We call $(\mu, e)$ a
\emph{witness} for $\trace\models\query$.

The \emph{model set} of a query $\query$ w.r.t. to a type set 
$\TSalt\subseteq\TS$ is 
    $\Mods{\Delta}{\query}\deff\{ \trace\in\Delta^+\,:\,\trace\models\query\}$.
Note that there exist \swggqueries $\query$ such that 
    $\mods{\query}=\varnothing$.
They have in common that either their generalised gap-size constraints 
conflict with the global window size or some gap-size constraints are in 
conflict among themselves. Lemma~\ref{cond:swggquery-compatibility}
characterises \swggqueries{} with compatible constraints.

\begin{Example}\label{exp:swggquery-unsatisfiable}
    Let $\TS=\{\ta,\tb\}$. Let $\query_1=(s,\ws,\wC_1)$ and 
    $\query_2=(s,\ws,\wC_2)$ be the two \swggqueries over $\VS$, where 
    $s = \ta \ta \ta \ta \ta \ta$ and $w=10$ for both queries, 
    and $C_1 = ( (7,7,3)_1, (6,6,3)_2, (0,0,1)_5 )$ and 
    $C_2 = ( (4,4,5)_1, (2,5,2)_3 )$.
    \par 
    \vspace*{0.5em}
    \hspace*{-2.1em}
    \begin{minipage}{0.45\textwidth}
        \footnotesize
        \addtolength{\tabcolsep}{-1pt}
        \begin{tabular}{ccccccccccccccc}
            &&&$1$& &$2$& &$3$& &$4$& &$5$& &$6$&\\
            $s$ &=& &$\ta$& &$\ta$& &$\ta$& &$\ta$& &$\ta$& &$\ta$&\\
            \multicolumn{4}{c}{} &
            \multicolumn{5}{c}{\upbracefill} &
            \multicolumn{3}{c}{} &
            \multicolumn{1}{c}{\upbracefill} &\\
            $\wC_1$ &:&
            \multicolumn{2}{c}{} &
            \multicolumn{5}{c}{$(7,7,3)_1$} &
            \multicolumn{3}{c}{} &
            \multicolumn{1}{c}{$(0,0,1)_5$} &\\
            \multicolumn{6}{c}{} & 
            \multicolumn{5}{c}{\upbracefill} &\\
            \multicolumn{6}{c}{} & 
            \multicolumn{5}{c}{$(6,6,3)_2$} &\\
            \multicolumn{3}{c}{} &
            \multicolumn{11}{c}{\upbracefill} &\\
            $\ws$ &:&
            \multicolumn{1}{c}{} &
            \multicolumn{11}{c}{10} & \\
        \end{tabular}
    \end{minipage}
    \hspace*{3em}
    \begin{minipage}{0.45\textwidth}
        \footnotesize
        \addtolength{\tabcolsep}{-1pt}
        \begin{tabular}{ccccccccccccccc}
            &&&$1$& &$2$& &$3$& &$4$& &$5$& &$6$&\\
            $s$ &=& &$\ta$& &$\ta$& &$\ta$& &$\ta$& &$\ta$& &$\ta$&\\
            \multicolumn{8}{c}{} &
            \multicolumn{3}{c}{\upbracefill} & \\
            $\wC_2$ &:&
            \multicolumn{6}{c}{} &
            \multicolumn{3}{c}{$(2,5,2)_3$} & \\
            \multicolumn{4}{c}{} & 
            \multicolumn{9}{c}{\upbracefill} &\\
            \multicolumn{4}{c}{} & 
            \multicolumn{9}{c}{$(4,4,5)_1$} &\\
            \multicolumn{3}{c}{} &
            \multicolumn{11}{c}{\upbracefill} &\\
            $\ws$ &:&
            \multicolumn{1}{c}{} &
            \multicolumn{11}{c}{10} & \\
        \end{tabular}
    \end{minipage}
    \phantom{text}\\
    A shortest trace over $\TS$ which satisfies $\wC_1$ is
        $\trace = \ta\ \tb\ \ta\ \tb\ \tb\ \tb\ \tb\ \ta\ \ta\ \ta\ \ta$.
    But $\trace\not\models\query_1$ since $\trace$ does not satisfy 
    $\ws=10$. Since $\trace$ is a shortest trace there exists no 
    trace satisfying both, $\ws$ and $\wC$.
    (The shortest trace is not unique since the sequences of $\tb$s could be 
    replaced by arbitrary types from $\TS$.)
    \par 

    Note that $\mods{\query_2}=\varnothing$ holds as well since $(4,4,5)_1$ and
    $(2,5,3)_3$ are incompatible: 
    For each trace $\trace$, substitution $\mu$ and embedding $e$ such that 
    $\mu(s) \subseq_e \trace$ and $e$ satisfies $(4,4,5)_1$, the second gap-size
    constraint is not satisfied, since it demands at least one further type 
    between $e(3)$ and $e(4)$, contradicting $(4,4,5)_1$.
\end{Example}

\newcounter{Counter_Thm:Lemma_cond:swggquery}
\setcounter{Counter_Thm:Lemma_cond:swggquery}{\value{theorem}}

\begin{Lemma}\label{cond:swggquery-compatibility}\label{label:swggquery-satisfiablity}
    An \swggquery $\query=(s,\ws,\wC)$ (over $\VS$ and $\Gamma$) is satisfiable, i.e. $\mods{\query} \neq \varnothing$, iff  there are no two sequences 
        \begin{align*}  
            C'  &=\left( 
                (c_1'^-,c_1'^+,r'_1)_{j'_1 }, \ 
                (c_2'^-,c_2'^+,r'_2)_{j'_2 }\ ,
	            \ldots, 
                (c_{|C'|}'^-,c_{|C'|}'^+,r'_{|C'|})_{j'_{|C'|}}\right) 
	    \end{align*}
	    and
        \begin{align*} 
            C'' &=\left( (c_1''^-,c_1''^+,r''_1)_{j''_1 }, \ 
                (c_2''^-,c_2''^+,r''_2)_{j''_2 }\ ,
                \ldots,  
                (c_{|C''|}''^-,c_{|C''|}''^+,r''_{|C''|})_{j''_{|C''|}}\right) 
	    \end{align*}
	    from $\wC \cup \set{ (0,\infty,1)_1, \ldots, (0,\infty,1)_{|s|-1} }$ 
        where 
        $j'_1 = j''_1$,  
        $j'_{|C'|}+r'_{|C'|} = j''_{|C''|}+r''_{|C''|}$, and 
        $j'_{i+1} = j'_i+r'_i$ and $j''_{i+1} = j''_i+r''_i$ 
        for all $i\in[|C'|-1]$, 
        with
	    \begin{quote}
	        \begin{enumerate}[(i)]
                \item $|s| + \sum\limits_{i=1}^{|C'|} c'^-_i-r_i'+1 \quad > \quad \ws$ \qquad, or
                \item $ \sum\limits_{i=1}^{|C'|} c'^-_i-r_i'+1 \quad > \quad \sum\limits_{i=1}^{|C''|} c''^+_i-r_i''+1$.  
	        \end{enumerate}	
        \end{quote} 
        \vspace*{-1em}
\end{Lemma}
For the rest of this paper, we only focus on queries with a non-empty model set.

\subsection{Syntax and semantics of \dswgqueries}\label{subsec:dswgqueries} 
\begin{Definition}\label{def:dswgquery}
    A \dswgquery $\query=(s,\ws,c)$ (over $\VS$ and $\Gamma$) is specified by
	\vspace*{-0.5em}
    \begin{itemize}
	\item a query string $s = s[1] \ldots s[\ell]\text{, whereby } s[i] = \begin{cases}
		x\in\VS &\\
		\dc \in \Pfinplus(\TS)
	\end{cases}$, 
	\item a global window size $\ws \in \natpos \cup \inftyset$ with 
        $w \geqslant |s|$ and
	\item a tuple $\wc = (\wc_1, \wc_2, \ldots, \wc_{|s|{-}1})$ 
of local gap-size constraints
(for $|s|$ and $w$), where $\wc_i =
(\wc^-_i, \wc^+_i) \in \mathbb{N} \times (\mathbb{N} \cup
\{\infty\})$, such that 
$\wc^-_i \leq \wc^+_i$ for every $i \in [|s|{-}1]$ and
$|s| {+} \sum^{|s|}_{i = 1}\wc^-_i\leq \ws$.
    \end{itemize} 
\end{Definition}
Note that setting all gap-size constraints of a \dswgquery $\query$ to 
$(0,\infty)$ corresponds to a query without gap-size constraints.
\par

The semantics of \dswgqueries is defined as follows: Again, variables in $s$ 
represent an arbitrary type, and each set $\dc$ stands for a disjunction.
Intuitively, a trace $\trace$ matches a query $\query=(s,\ws,\wc)$ if the 
variables in $s$ can be replaced by types and each occuring of an $\dc$ can 
be mapped to a single type $\type\in\dc$, such that the resulting string $s'$ 
occurs as a subsequence in $\trace$ that spans at most $\ws$ types and the 
gap between $s'[i]$ and $s'[i{+}1]$ in $\trace$ has length between $\wc_i^-$ 
and $\wc_i^+$, for all $i<\ell\deff |s|$.
\par

An alternative description of these semantics, which will be more
convenient for our formal proofs, involves a bit more notation:
We say that an embedding \emph{$e : [\ell] \to [n]$ satisfies} a
global window size $\ws$, if $e(\ell)-e(1)+1 \leq \ws$;
and we say that \emph{$e$ satisfies} a tuple $\wc = (\wc_1, \wc_2,
\ldots, \wc_{\ell{-}1})$ of local gap-size constraints (for $\ell$ and
$\ws$), if $\wc^-_i \ \leq \ e(i{+}1){-}1-e(i) \ \leq \ \wc^+_i$ \ for
all $i<\ell$.
\par

A \emph{substitution of size $\ell$}
is a mapping $\mu_\ell : ([\ell] \times \VS \cup \Pfinplus(\TS)) \to (\VS \cup \Pfinplus(\TS))$ with:
\[ \mu_\ell(i, z) = \begin{cases} x \in(\VS \cup \Pfinplus(\TS)) &, i=1 \text{ and } z \in \VS \\
								\mu_\ell(1,z) &, i>1 \text{ and } X \in \VS \\
								z' &, \varnothing \subset z' \subseteq z \text{ for all non-empty } z \finsubseteq \TS. 	 \end{cases}\]

\noindent We  extend substitutions of size $\ell$ to mappings 
	$([\ell] \times \VS \cup \Pfinplus(\TS))^+ \to (\VS \cup \Pfinplus(\TS))^+$ 
for strings $s \in (\VS \cup \Pfinplus(\TS))^+$ of size $\ell$
in the obvious way, i.e., 
$\mu(s) = \mu_\ell(1,s[1]) \mu_\ell(2,s[2]) \cdots \mu_\ell(\ell,s[\ell])$. Since the size of the string must match the size of the substitution, we can omit the index $\ell$ if we apply it to a string.
Particularly, we can omit the parameter $i$ for the position if the second parameter is a variable, as we have for all variables $z$ that $\mu(i,z) = \mu(i',z)$ for all $i,i' \in [l]$, or if the position is given by the context, i.e. we write $\mu(s[i])$ instead of $\mu(i,s[i])$.

A \dswgquery $\query=(s,\ws,\wc)$ \emph{matches in a trace $\trace\in\TS^+$}
(or, $\trace$ \emph{matches} $\query$, in symbols: $\trace\models\query$),
if and only if there are 
a substitution $\mu : ([\ell] \times\VS \cup \Pfinplus(\TSalt)) \to \TS$ 
(i.e., there are only singeltons in the 
co-domain of $\mu$ and every type of 
$\mu(s)$ is the unique element of its singleton) 
and an embedding $e : [|s|] \to [|t|]$ that satisfies $\ws$ and $\wc$,
such that $\mu(s) \subseq_e t$. We call $(\mu,e)$ \emph{a witness for $\trace\models\query$}.

\begin{Example}
	Let $x_1, x_2, x_3\in\VS$ and $\TS = \{\ta, \tb, \tc\}$. We consider a query $\query = (s, \ws,
	\wc)$, where $s = x_1 \{\ta,\tb\} x_1 x_2 \{\tc\} x_3 \{\ta,\tc\} x_1$, $w = 25$ and
	$\wc = ((0, 1), (2, \infty),$ $(3, \infty), (0, 5),
	(0, 5)$, $(1, 5), (1, 2))$. For $t_1,t_2 \in \Gamma^*$ we consider the trace 
	$\trace = \trace_1 \tc \ta \tb \tb \tc \ta \tb
	\ta \tc \ta \tb \ta \tc \tb \tc \tb \tb \ta \tc \trace_2$.
	We observe that $\trace \models \query$, and a
	witness substitution and embedding can be illustrated as follows: 
	\footnotesize
	\begin{align*}
		&s& &=& && &x_1& &\{\ta,\tb\}& && &x_1& &&           &x_2& &\{\tc\}& &x_3& && &\{\ta,\tb\}& && &x_1\,,& &&\\
		&\trace& &=& &\trace_1& &\tc&  &\hspace*{0.75em}\ta& &\tb \tb&      &\tc& &\ta \tc \ta& &\ta&   &\hspace*{0.35em}\tc& &\tb&  &\tc&     &\hspace*{0.75em}\tb& &\tb &
		&\tc& &\trace_2\ .&
	\end{align*}
\end{Example}

\noindent We close this subsection with two little observations. First, it is reasonable to assume that $\dc$ is a proper subset of \TS, otherwise we could also use a wildcard instead of a disjunction. Second, in the case of a finite alphabet $\Gamma$, we obtain the possibility to express a simple kind of negation. We can build a query where $s$ contains a substring $s'= \ta \dc \tb $ with $\dc=\Gamma\setminus\set{c}$ and corresponding conditions $c=((0,0),(0,0))$ that is only matched by traces where in between of the relevant $\ta$ and $\tb$ occurs exactly one letter that is not $\tc$. Unfortunately, it is not possible to express negation in general, so we cannot express the following: $\ta$ and $\tb$ have one or two letters in between, none is $\tc$.

\subsection{About containment}\label{subsec:containment} 

This section is dedicated to a characterisation of containment, a classical 
property considered in database theory. 
\par 

An \swggquery \query\ is called an
\emph{$(\ell,\ws,\wC)$-\swggquery}  (over $\VS$ and $\TS$) if 
	$\query = (s,\ws, \wC)$ with $|s| = \ell$, 
\emph{$(\ell,\ws,\wc)$-\dswgqueries} are analogously defined. 
The para- meter $\ell$ will be called \emph{string length}. If the maximal size 
of typesets occuring in an $(\ell,\ws,\wc)$-\dswgquery $\query$ is bounded by a 
number $k\geq 1$, we call $\query$ an $(\ell,\ws,\wc,k)$-\dswgquery 
(for \swggqueries $k$ always equals $1$).
\par 

Given an \swggquery{} or \dswgquery{} $\query$ we omit the prefix and call
$\query$ a \emph{query}, if $\query$ may be both, or it is clear from the 
context whether $\query$ is an \swggquery{} or \dswgquery.
We use \emph{$(\ell,\ws,\WC)$-query} (or \emph{$(\ell,\ws,\WC,k)$-query}) as 
notation for queries which might be \swggquery{} \emph{or} \dswgquery and assume 
that the gap-size constraints $\WC$ are compatibile with $\ell$ and $\ws$ or 
satisfy Lemma~\ref{cond:swggquery-compatibility}, respectively.

We write $\types(\query)$ (or $\types(s)$), $\typesets(\query)$ 
(or $\typesets(s)$) and $\vars(\query)$ (or $\vars(s)$) for the set of types, 
the set of all typesets and the set of variables, respectively, that occur in 
$\query$'s query string $s$. I.e., 
\begin{align*}
	\types(\query)&\deff\{\gamma \in \TS \,|\, \text{ there ex. }i\in[|s|]: \gamma \in s[i] \}\\ 
	\typesets(\query)&\deff\{\dc\subseteq\TS \,|\, \text{there ex. }i\in[|s|]: s[i]=\dc\}\\
	\vars(\query)&\deff\{x\in\VS \,|\, \text{there ex. }i\in[|s|]: s[i]=x\}
\end{align*}

\noindent
For the reason of readability we omit braces in query strings if the set 
consists only of an unique element. Therefore, we write for example 
	$s = x_1\ta\tb x_2\ta\set{\ta,\tc}\tb$ 
instead of $s = 
	x_1\set{\ta}\set{\tb}x_2\set{\ta}\set{\ta,\tc}\set{\tb}$. 
Vice versa, we can consider a query string $s$ over $\VS \cup \Gamma$ as a 
string over $\VS \cup \P_1(\TS)$ where every 
$s[i]$ is a singleton.

A query $\query$ is said to be \emph{contained} in a query $\query'$ w.r.t. to
a set $\TSalt\subseteq\TS$ (we write $q \contResp{\Delta} q'$) if 
    $\Mods{\Delta}{\query} \subseteq \Mods{\Delta}{\query'}$.

\begin{Definition}\label{def:swggquery-homomorphism}\label{def:homomorphism}
A homomorphism from  $\query'$ to $\query$ is a substitution $h$ such that $h(s') = s$ and the following property holds:
\vspace*{-0.75em}
\begin{quote}
For every $z \in \Vars$ that occurs at least twice in the query string $s'$ of 
$\query'$ and is mapped to a subset of $\Gamma$ via $h$, we have $h(z)$ is a 
singelton.
\end{quote}
 We write $\query'\Hom\query$ to express that there exists a homomorphism
	from $\query'$ to $\query$.
\end{Definition} 
 
\noindent This additional property of homomorphisms feels arbitrary or artificial, but it is perfectly tailored to our discovery algorithm and if the considered class of queries in Definition~\ref{def:homomorphism} is the class of all $(\ell,\ws,\wC)$-\swggquery, then in any way, we have $s[i]$ is a singelton for all $i\in[\ell]$. Now, the following theorem gives a characterisation of containment.

\newcounter{Counter_Thm:queryContVsHom}
\setcounter{Counter_Thm:queryContVsHom}{\value{theorem}}
\begin{theorem}\label{thm:queryContVsHom}
Given some sufficiently large $\TS$.
Let $\query$ and $\query'$ be $(s,\ws,\WC)$-queries over $\VS$ and $\TS$.
If $\query$ and $\query'$ are satisfiable, it holds, that:
\vspace*{-0.75em}
$$ \query \contResp{\Gamma} \query' \iff \query'\Hom\query.$$ 
\end{theorem}

\noindent Sufficiently large, in the context of Theorem~\ref{thm:queryContVsHom} 
means, $|\Gamma|\geq 2$ for the case of \dswgqueries. 
In the case of \swggqueries 
the size of $\TS$ depends on gap-size constraints with range greater than $1$. 
Intuitively, the necessary size of $\TS$ depicts how much structure information 
of $\query'$ can be hidden in the gaps of $q$. For further information and an 
example consider the theorems proof from page \pageref{Appenix:homtheorem} 
onwards.

\subsection{Correlation to \swgqueries}\label{subsec:swgqueries} 
Note that an \swggquery query $\query=(s,\ws,\wC)$ with 
$\wC=\{c_1,\ldots,c_{\ell-1}\}$ and $r_i=1$ for all $i\in[\ell-1]$
precisely corresponds to the notion of \swgqueries introduced in
\cite{KSSSW22}.
Let $\query$ be an $(\ell,\ws,\wc,k)$-query containing typesets over 
$\Pfinplus(\TS)$. If 
    $k=1$
this correpsonds to the syntax and semantics of \swgqueries as well.
\par 
In \cite{KSSSW23} a mapping between one-dimensional and multi-dimensional 
sequence data was introduced, such that a multi-dimensional trace matches a 
multi-dimensional query if and only if the corresponding one-dimensional trace 
matches the corresponding one-dimensional query. This mapping can be adapted to
\swggqueries{} and \dswgqueries{}. 

\section{Discovery}\label{sec:discovery-short}
The question of how meaningful \swggqueries{} and \dswgqueries{} can be 
discovered from a given set of traces is of peculiar interest and was 
answered algorithmically in \cite{KSSSW22} for \swgqueries.
We adapt these results to \swggqueries{} and \dswgqueries{}.
\par 
A \emph{sample} is a finite, non-empty set $\Sample$ of traces over $\TS$.
Given a sample $\Sample$, let $\TS_\Sample$ be the set of all types occurring in 
$\Sample$, i.e. $\bigcup_{\trace\in\Sample} \types(\trace)$.
The \emph{support} $\Supp(\query,\Sample)$ of a query $\query$ in $\Sample$
is defined as the fraction of traces in the sample that match $\query$, i.e.
    $\Supp(\query,\Sample)\deff
        \frac{|\{\trace\in\Sample\ :\ \trace\models\query\}|}{|\Sample|}$.
A \emph{support threshold} is a rational number $\supp$ with $0<\supp\leq 1$.
A query $\query$ is said to \emph{cover} a sample $\Sample$ \emph{with support}
$\supp$ if $\Supp(\query,\Sample)\geq\supp$.
Let $\Sample$ be a sample, $\supp$ be a support threshold and 
    $k\in[|\TS_\Sample|-1]$.
An $(\ell,\ws,\WC,k)$-query $\query$ is called \emph{descriptive for $\Sample$ 
w.r.t $(\supp,(\ell,\ws,\WC,k))$} if 
    $\Supp(\query,\Sample)\geq \supp$, 
and there is no other $(\ell,\ws,\WC,k)$-query $\query'$ with 
    $\Supp(\query',\Sample)\geq \supp$
and
    \query' \propContResp{\Gamma} \query.
A type $\type\in\TS$ (or a typeset $\dc \in \Pfinplus(\TS)$) \emph{satisfies $\supp$ 
w.r.t. to $\Sample$}, if the fraction of traces containing $\type$ (or some 
$\type\in\dc$) is greater than or equal to $\supp$. The set of all types (or 
typesets) that satisfy $\supp$ w.r.t. to $\Sample$ is 
\begin{align*}
    \TSalt(\Sample,\supp) \deff \{\dc\in\Pfinplus(\TS)\ :\ 
        \frac{
            |\{\trace\in\Sample\ :\ \text{ex. } \type\in \dc \text{ s.t. }
            \type\in\types(\trace)\}|}
        {|\Sample|} \geq\supp
        \}.
\end{align*}
We omit $\Sample$ and $\supp$, if they are clear from the context.
For $i\in[|\TS_\Sample|-1]$ we write $\TSalt_i$ to denote the subset of 
$\TSalt$ which contains only typesets of size $i$. For a descriptive
$(\ell,\ws,\WC,k)$-query $\query$, 
    $\typesets(\query)\subseteq \TSalt_1 \dot\cup \dots \dot\cup\TSalt_k$
holds.
This corresponds to
    $\TSalt = \TSalt_1 = \{\type\in\TS:
         \frac{|\{\trace\in\Sample\ :\ \type\in\types(\trace)\}|}{|\Sample|}
         \geq\supp\}
    $
if the considered query is an \swggquery.
\par 
Given an $(\ell,\ws,\WC,k)$-query $\query=(s,\ws,\WC)$ and a symbol $z$ from 
$\VS \cup \Pfinplus(\TS)$ we let 
    $\Pos{\query}{z} = \Pos{s}{z} = \{i\in[\ell]\ :\ s[i]=z\}$
be set of all positions $i$ in $s$ that carry $z$. 
Given a set of positions $P\subseteq[\ell]$, and a symbol 
    $z\in\VS \cup \Pfinplus(\TS)$
we write 
    $\replace{s}{P}{z}$ 
to denote the query string $s'$ which is obtained from $s$ by setting $s[i]$ to 
$z$, for all $i\in P$. 
Let $x\in\VS$. We write $\replace{s}{x}{z}$ as an abbreviation for 
    $\replace{s}{\Pos{\query}{x}}{z}$.
Next, we present the algorithmical idea for 
query discovery:
\par \vspace*{0.5em}
\textbf{Compute Descriptive Query Problem} ($\compDescProb$):
On input of a sample $\Sample$ over $\TS$, a support threshold $\supp$, 
a string length $\ell\in\nat$, a global window size $\ws\geq\ell$, 
a tuple $\WC$ of gap-size constraints, and $k\in[|\TS_\Sample|-1]$,
the task is to compute an $(\ell,\ws,\WC,k)$-query $\query$ that is
descriptive for $\Sample$ w.r.t. $(\supp,(\ell,\ws,\WC,k))$.
\par \vspace*{0.5em}
Pseudocode of an algorithm solving $\compDescProb$ is provided in 
Algorithm~\ref{algo:ComputeDescQueryFromq}. 
Given $\Sample$, $\supp$ and query parameters $(\ell,\ws,\WC,k)$ 
as input, the algorithm first builds the 
\emph{most general query} 
    $\query=\mgq$ 
for $(\ell,\ws,\WC,k)$. Its query string consists of $\ell$ distinct variables, 
i.e. $\mgs = x_1 \dots x_\ell$, and $\mgq$ is most general in the sense that 
    $\query' \contResp{\TS} \mgq$
for each $(\ell,\ws,\WC,k)$-query $\query'$. 
\begin{algorithm}[h!tbp]
	\SetAlgoNoEnd
    \SetAlCapNameFnt{\small}
    \SetAlCapFnt{\small}
	\LinesNumbered
	\DontPrintSemicolon
	\SetKwInOut{Input}{Input}
	\SetKwInOut{Output}{Returns}
	\Input{sample $\Sample$; support threshold $\supp$ with $0<\supp\leq
		1$; $(\ell,\ws,\WC,k)$ 
        }
	\Output{descriptive query $\query$ for $\Sample$ w.r.t.\ 
		$(\supp,(\ell,\ws,\WC,k))$ or error message $\Error$}
    $s := \mgs$;\ \ $q:=(\mgs,\ws,\WC)$
        \label{algo:MgqDefLine}
        \tcp*{$\scriptstyle\text{query string and query}$}
    \lIf{$\Supp(\query,\Sample) < \supp$\label{algo:ErrorMsgLineBeginning}}
    {\textbf{\textup stop and} \Return{$\Error$}}
    $\TSalt:= \TSalt_1 \dot\cup \dots \dot\cup \TSalt_k$
        \label{algo:TSaltDefLine}
        \tcp*{$\scriptstyle\text{typesets to be considered}$}
	$\NR:=\vars(q)$;\ \ $\AV:=\varnothing$
        \tcp*{$\scriptstyle\text{unvisited and available variables}$}
    \While{$\NR \neq \varnothing$\label{algo:mainLoopLine}}{
		\textbf{select} an arbitrary $x\in \NR$ and let 
            $\NR:=\NR\setminus\set{x}$ and
            $\TSalt_1 := \TSalt_1\cup\AV$
		    \label{algo:selectNextPositionLine}\;
        \For{$i=1$ to $k$\label{algo:forDeltaiLoopLine}}{            
            $\textup{replace}:=\textsf{False}$\;
            \label{algo:replaceToFalseLine}
            \While{$\TSalt_i\neq\varnothing$\label{algo:innerLoopLine}}{
                \textbf{select} an arbitrary $y\in\TSalt_i$ and 
                let $\TSalt_i :=\TSalt_i\setminus\set{y}$
                    \label{algo:selectNextReplacementLine}\; 
                $\query' := (\replace{s}{x}{y},\ws,\WC)$\;
                \label{algo:ReplaceForCheckLine} 
                \uIf{$\Supp(\query',\Sample) \geq \supp$
                    \label{algo:requestsToOracleLine}}{
                        $s:= \replace{s}{x}{y}$;\ \ $\textup{replace}:=\textsf{True}$
                            \tcp*{$\scriptstyle\text{ReplaceOp}$}
                            \label{algo:ReplaceOpLine} 
                        \textbf{break} for loop
                }
		    }
        }
		\lIf{\textup{replace} is \textsf{\upshape False}}{
            $\AV:= \AV\cup\set{x}$
            \tcp*[f]{$\scriptstyle\text{NoChangeOp}$}
            \label{algo:noChangeOpLine}
		}
	}
	\textbf{\textup stop and}\label{algo:queryoutputline}
    \Return{$q:=(s,\ws,\WC)$}\;
    \caption{$\QComputeDescrQueryFromq$}
	\label{algo:ComputeDescQueryFromq}
\end{algorithm}
If $\Supp(\query,\Sample)<\supp$ the algorithm stops and returns $\Error$
(line~\ref{algo:ErrorMsgLineBeginning}), because no other query $\query'$ with  
    $\query'\contResp{\TS}\query = \mgq$
can describe $\Sample$ with sufficient support.
\par 
Otherwise, the algorithm searches for an admissable \emph{replacement operation}
for each variable $x\in\NR\deff\vars(s)=\{x_1,\ldots,x_\ell\}$ during the main 
loop (Line~\ref{algo:mainLoopLine}). 
A replacement operation replaces $x$ by an element $y\in\TSalt_i$ (during the 
$i$-th iteration of the for-loop) which may be a typeset or an 
\emph{available} variable $y\in\AV$ (if i=1). 
The replacement operation is stored in 
$\query$ and called \emph{admissable} if the resulting query satisfies the 
support threshold (lines~\ref{algo:ReplaceForCheckLine}--\ref{algo:ReplaceOpLine}). 
If $\Supp(\replace{}{\Pos{\query}{x}}{y},\Sample)<\supp$ for all $y\in\TSalt_i$
and all $i\in[k]$ 
the query string remains unchanged and $x$ gets available 
(line~\ref{algo:noChangeOpLine}).
After each variable in $\vars(\mgs)$ has been considered, the algorithm
terminates and produces the current query as output 
(line~\ref{algo:queryoutputline}).
\par 
Next we depict an exemplaric run of Algorithm~\ref{algo:ComputeDescQueryFromq}. 
We refer to the appendix for a brief discussion, why 
$\TSalt$ is passed through incrementally in Line~\ref{algo:forDeltaiLoopLine}.
\newcounter{Counter_Exa:Run}
\setcounter{Counter_Exa:Run}{\value{theorem}}
\begin{Example} \label{examp:run}
    Let $\TS=\{\ta,\tb,\tc\}$ and $x_1, x_2, x_3\in\VS$. Consider the sample    
    $\Sample=\{\ta\tb\tb,\ta\tc\tc\}$, $\supp=1.0$, $\ell=\ws=3$, 
    $\wc = ((0,0),(0,0))$ and $k=2$.
    \par 
    On input $(\Sample,\supp,(\ell,\ws,\wc,k))$ the algorithm first 
    generates $\query = (x_1 x_2 x_3,\ws,\wc)$. Since $\query$
    satisfies the support threshold the algorithm proceeds by computing 
        $\TSalt =   \{ \{\ta\}\} \dot\cup 
                    \{\{\ta,\tb\}, \{\ta,\tc\}, \{\tb,\tc\} \}$.
    Assume the algorithm selects $x:=x_3$ during the first iteration of the 
    main loop. It turns out that $\TSalt_1 = \{ \{\ta\}\}$ does not contain a 
    typeset for an admissable replacement of $x_3$. Hence, the algorithm 
    considers $\TSalt_2$ during the second transition of the for-loop in 
    Line~\ref{algo:forDeltaiLoopLine}. The only admissable replacement is 
        $\replace{s}{x_3}{\{\tb,\tc\}}$, 
    and $s$ is replaced by $x_1 x_2 \{\tb,\tc\}$ ($\AV$ remains empty).
    \par 
    Let us assume that during the second transition through the main loop the
    algorithm selects $x:=x_1$ and $y:=\{\ta\}\in\TSalt_1$.
    The replacement of $x_1$ by $\{\ta\}$ is admissible 
    (as it has support 1 on $\Sample$). Therefore, $s$ is replaced by 
        $\{\ta\} x_2 \{\tb,\tc\}$
    and $\AV$ remains unchanged again.
    \par
    In its last iteration (during the second transition through the for-loop),
    $\replace{s}{x_2}{\{\tb,\tc\}}$ is the only admissible replacement 
    operation. The run terminates after this iteration and outputs the query 
    $\query=(s,\ws,\wc)$ with $s=\{\ta\} \{\tb,\tc\} \{\tb,\tc\}$.
\end{Example}
\begin{figure}[h!]
    \begin{center}
        \begin{tikzpicture}
            \node[] (start) at (0,0.99) {};
            \node[] (a) at (-3,1) {$\{a\}$};
            \node[] (b) at (-1,1) {\textcolor{blue}{$\{b\}$}};
            \node[] (c) at (1,1) {\textcolor{blue}{$\{c\}$}};
            \node[] (d) at (3,1) {$\{d\}$};
            \node[] (ab) at (-4,2) {\textcolor{blue}{$\{a,b\}$}};
            \node[] (ac) at (-2.5,2) {\textcolor{blue}{$\{a,c\}$}};
            \node[] (ad) at (-1,2) {$\{a,d\}$};
            \node[] (bc) at (1,2) {\textcolor{blue}{$\{b,c\}$}};
            \node[] (bd) at (2.5,2) {\textcolor{blue}{$\{b,d\}$}};
            \node[] (cd) at (4,2) {\textcolor{blue}{$\{c,d\}$}};
            \node[] (abc) at (-3,3) {\textcolor{gray}{$\{a,b,c\}$}};
            \node[] (abd) at (-1,3) {\textcolor{gray}{$\{a,b,d\}$}};
            \node[] (acd) at (1,3) {\textcolor{gray}{$\{a,c,d\}$}};
            \node[] (bcd) at (3,3) {\textcolor{gray}{$\{b,c,d\}$}};
            \node[] (abcd) at (0,3.75) {\textcolor{gray}{$\{a,b,c,d\}$}};
            \path   (a) edge (ab)
                    (a) edge (ac)
                    (a) edge (ad)
                    (b) edge (ab)
                    (b) edge (bc)
                    (b) edge (bd)
                    (c) edge (ac)
                    (c) edge (bc)
                    (c) edge (cd)
                    (d) edge (ad)
                    (d) edge (bd)
                    (d) edge (cd)
                    (ab) edge [gray] (abc)
                    (ab) edge [gray] (abd)
                    (ac) edge [gray] (abc)
                    (ac) edge [gray] (acd)
                    (ad) edge [gray] (abd)
                    (ad) edge [gray] (acd)
                    (bc) edge [gray] (abc)
                    (bc) edge [gray] (bcd)
                    (bd) edge [gray] (abd)
                    (bd) edge [gray] (bcd)
                    (cd) edge [gray] (acd)
                    (cd) edge [gray] (bcd)
                    (abc) edge [gray] (abcd)
                    (abd) edge [gray] (abcd)
                    (acd) edge [gray] (abcd)
                    (bcd) edge [gray] (abcd);                
        \end{tikzpicture}
    \end{center}
    \caption{
        Let $\TS=\{\ta,\tb,\tc,\td\}$, 
        $\Sample = \{\tc\ta\tb\tb\tc\ta\tc\tb,\ \tc\tb\tb\tb\ta\tc\tc\tb,
        \tc\tc\tb\tb\tc\tc\tc\tb\}$,
        $\supp=1.0$ and $k=2$.
        Depicted is a top-down walk through $\Pfinplus(\TS)$
        for $\Sample$, starting with typesets of size $k=2$. The typesets marked
        in blue represent $\TSalt$.
        }
    \label{fig:TSaltCal-FullTree}
\end{figure}
Note that algorithm~\ref{algo:ComputeDescQueryFromq} computes an \swggquery{}
if $k=1$ and generalised gap-size constraints are given. Furthermore, during 
each iteration of the main loop, the for-loop is transisted only once.
It remains to discuss how $\TSalt$ can be calculated in case that $k>1$.
Starting with all subsets $\dc$ of $\Pfinplus(\TS)$ with 
    $|\dc|=k$
it suffices to explore $\Pfinplus(\TS)$ in a top-down manner: 
we walk through the search space level-wise and check whether the 
current typesets satisfy $\supp$ w.r.t. $\Sample$. If this is not the case 
for a typeset $\dc$, all typesets $\dc'\subset\dc$ can be deleted from the 
search space, since they do not satisfy $\supp$.
An example is depicted in Figure~\ref{fig:TSaltCal-FullTree}.

\newcounter{Counter_Thm:Thm:discovery}
\setcounter{Counter_Thm:Thm:discovery}{\value{theorem}}
\begin{theorem}\label{thm:ComputeDescrQuery}
    Given some sufficiently large $\TS$.
    Let $\Sample$ be a sample, let $\supp$ be a support threshold with 
    $0<\supp\leq 1$, let $(\ell,\ws,\WC,k)$ be query parameters 
    with $k=1$ if $\WC=\wC$.
	\begin{enumerate}[(a)]
		\item
            If there does not exist any $(\ell,\ws,\WC,k)$-swgg-- or dswg--query 
            that is descriptive for $\Sample$ w.r.t.\ $(\supp,(\ell,\ws,\WC,k))$
            then there is only one run of
            Algorithm~\ref{algo:ComputeDescQueryFromq} upon the defined input,
            and it stops in line~\ref{algo:ErrorMsgLineBeginning} with output 
            $\Error$.
		\item
		    Otherwise, every run of Algorithm~\ref{algo:ComputeDescQueryFromq} 
            upon input $(\Sample,\supp,(\ell,\ws,\WC,k))$ terminates and 
            outputs an \swggquery{} or \dswgquery{} $\query$ (depending on $k$), 
            with $|\dc|\leq k$ for all $\dc\in\typesets(\query)$, that is 
            descriptive for $\Sample$ w.r.t.\ $(\supp,(\ell,\ws,\WC,k))$.
	\end{enumerate}
    \vspace*{-1em}
\end{theorem}

\noindent We refer to the appendix for the full proof. 
Analysing the complexity of the algorithm, identifies two bottle necks. 
First the $\Delta$-calculation in the case of \dswgqueries. This can be 
handeled by adjusting the parameter $k$, i.e. by bounding the size of the 
disjunctive clauses in the query string. 
The second is already known from \cite{KSSSW22} and is caused by the recurring
calls of a \emph{matching} subroutine. The refered results imply 
$\npclass$-hardness for our algorithm as we can use it as well for \swgqueries 
from \cite{KSSSW22}. Membership can be obtained by guessing a witness.
\vspace*{0.75em}

\section{Conclusion and Future Work}\label{sec:conclusion}
We model sequence data as traces and discover descriptive queries over traces 
to find a characteristic template for situations of interests. 
Since an increased expressive power of the underlying query language leads to 
a more detailed picture of sois, we extended \swgqueries, introduced in 
\cite{KSSSW22}, in two different ways. First, by generalising the gap size 
constraints (Section~\ref{subsec:swggqueries}) and second, by adding the 
possibilty of disjunctions (Section~\ref{subsec:dswgqueries}). 
We adopted and extended the discovery algorithm to our approach and ensured 
that the essential complexity properties are preserved 
(Section~\ref{sec:discovery-short}). 
Note that the extended approach can be applied to the multi-dimensional 
setting, analogously to \cite{KSSSW23}.

For future work we will merge both extensions to one query language. We are 
interested in a more general notion of disjunction and negation, and a more 
generous possibilty to describe gaps. For the latter \cite{DayEtAl2022} is a 
good yardstick. 
An in-depth (parameterised) complexity analysis is intended as well. Since the 
crucial point is the inherent complexity of the matching problem, we are 
working on data strcutures to improve the computation in practical application.
In the long run we will investigate containment for relaxed query parameters 
$(\ell,\ws,\wc, k)$.

\section*{Acknowledgments}
  We thank Markus L. Schmid for useful discussions.
  Sarah Kleest-Mei{\ss}ner was supported by the German Research Foundation 
  (DFG), CRC 1404: ``FONDA: Foundation of Workflows for Large-Scale 
  Scientific Data Analysis''.

\bibliographystyle{plain}
\bibliography{bibfile}

\appendix
\section*{APPENDIX}

This appendix contains technical details and further information which were 
omitted in the main part of the paper.
\begin{itemize}
    \item Appendix~\ref{app:satisfiabOfSWGG} considers satisfiability of \swggqueries and contains a proof of Lemma~\ref{cond:swggquery-compatibility}.
    \item Appendix~\ref{app:thm7} provides a proof of Theorem~\ref{Appenix:homtheorem} 
        including a detailed example discussing the requiered size of $\TS$.
    \item In Appendix~\ref{app:iso} the definitons of (partial) isomorphisms are
        given, which are crucial the proof provided in Appendix~\ref{app:discover}.
    \item Appendix~\ref{app:discover} provides detailed information on
        Example~\ref{examp:run}, i.e. a particular run of Algorithm~\ref{algo:ComputeDescQueryFromq},
        and a proof of Theorem~\ref{thm:ComputeDescrQuery}, stating that the discovery algorithm
        is correct.
\end{itemize}
\section{Regarding Satisfiability of \swggqueries}\label{app:satisfiabOfSWGG}

\newcounter{restoreAppTheorem} 

Let $\query=(s,\ws,\wC)$ be an \swggquery over $\TS$ and $\TSalt\subseteq\TS$.
We define \Mintrace{\TSalt}{\query} to be the set of all traces of minimal 
length matching $\query$, i.e. 
    $\Mintrace{\TSalt}{\query} \deff \{ \trace\in\TSalt^+ \ :\ 
        \trace\models\query \text{ and }
        \text{there exists no } \trace'\in\TSalt^+ \text{ with }  
            \trace'\models\query \text{ and } |\trace'|<|\trace|
    \}$
Given a trace $\trace\in\mintrace{\query}$, let $g_i\in\TS^\star$ be the 
\emph{gap string} of minimal length (according to $\wC$) between $s[i]$ and 
$s[i+1]$, for all $i\in[\ell-1]$. 
We observe that $\ws$ and $\wC$ can only be compatible if 
    $|s|+\sum_{i=1}^{\ell-1} |g_i| \leq \ws$.

Given two sequences $C'$ and $C''$ from $\wC$ such that the generalised gap-size 
constraints within $C'$ and $C''$ are non-overlapping. Intuitively speaking,
$C'$ and $C''$ can only be compatible if the induced gap strings (of minimal
length) of $C'$ do not contradict the upper bounds on gap strings induced by 
$C''$.
These observations can be formalised as follows and ensure
the satisfiablity of \swggqueries.

\setcounter{restoreAppTheorem}{\value{theorem}}
\setcounter{theorem}{\value{Counter_Thm:Lemma_cond:swggquery}}
\begin{Lemma}\textbf{(restated)}
    An \swggquery $\query=(s,\ws,\wC)$ (over $\VS$ and $\Gamma$) is satisfiable, i.e. $\mods{\query} \neq \varnothing$, iff  there are no two sequences 
        \begin{align*}  
            C'  &=\left( 
                (c_1'^-,c_1'^+,r'_1)_{j'_1 }, \ 
                (c_2'^-,c_2'^+,r'_2)_{j'_2 }\ ,
	            \ldots, 
                (c_{|C'|}'^-,c_{|C'|}'^+,r'_{|C'|})_{j'_{|C'|}}\right) 
	    \end{align*}
	    and
        \begin{align*} 
            C'' &=\left( (c_1''^-,c_1''^+,r''_1)_{j''_1 }, \ 
                (c_2''^-,c_2''^+,r''_2)_{j''_2 }\ ,
                \ldots,  
                (c_{|C''|}''^-,c_{|C''|}''^+,r''_{|C''|})_{j''_{|C''|}}\right) 
	    \end{align*}
	    from $\wC \cup \set{ (0,\infty,1)_1, \ldots, (0,\infty,1)_{|s|-1} }$ 
        where 
        $j'_1 = j''_1$,  
        $j'_{|C'|}+r'_{|C'|} = j''_{|C''|}+r''_{|C''|}$, and 
        $j'_{i+1} = j'_i+r'_i$ and $j''_{i+1} = j''_i+r''_i$ 
        for all $i\in[|C'|-1]$, 
        with
	    \begin{quote}
	        \begin{enumerate}[(i)]
                \item $|s| + \sum\limits_{i=1}^{|C'|} c'^-_i-r_i'+1 \quad > \quad \ws$ \qquad, or
                \item $ \sum\limits_{i=1}^{|C'|} c'^-_i-r_i'+1 \quad > \quad \sum\limits_{i=1}^{|C''|} c''^+_i-r_i''+1$ 
	        \end{enumerate}	
        \end{quote} 
\end{Lemma}
\setcounter{theorem}{\value{restoreAppTheorem}}

\begin{proof}[Proof (Sketch).]

First, it is easy to verify, that an \swggquery $\query=(s,\ws,\wC)$ with $$\wC \subseteq \set{ (0,\infty,1)_1, \ldots, (0,\infty,1)_{|s|-1} }$$ has a model by assigning every variable of $s$ to some arbitary element of $\Gamma$.
 
We observe that $C'$ and $C''$ are non-overlapping and connected sequences of conditions, that means,  taking a component $(c^-,c^+,r)_j$ of such a sequence, it speaks about the gaps between position $j$ and $r+j$ of the query string and the following condition connects seamlessly. So the sequence speaks about an entire (part) of the query string.  Note, that every sequence can be enriched by conditions of form $(0,\infty,1)$ to an non-overlapping and connected sequence speaking about the entire query string.

Next, for such a sequence $C$ the sum $\sum\limits_{i=1}^{|C|} c^-_i-r_i+1$ is the sum of all minimal gaps and  $\sum\limits_{i=1}^{|C|} c^+_i-r_i+1$ is the sum of all maximal gaps fulfilled by all models of the (part of the) query string.

Hence, obviously, the first inequality states that the minimal size of a model is larger than the global size $w$. And for the second inequality, we have that for (part of) the query string  the total number of gap filling letters in the model has to be bigger in the minimum than in the required maximum. Therefore, both subsets of the condition set of the query contradict each other.

Finally, we have, if there are two non-overlapping and connected sequences of conditions $C'$ and $C''$, such that $C'$ fulfills the first inequality or both fulfill the second inequality, than the model set of $\query$ is empty, i.e. $\mods{\query} = \varnothing$.  
\end{proof}

\section{About Homomorphisms}\label{app:thm7}

\setcounter{restoreAppTheorem}{\value{theorem}}
\setcounter{theorem}{\value{Counter_Thm:queryContVsHom}}
\begin{theorem}\textbf{(restated)} \label{Appenix:homtheorem}
Given some sufficiently large $\TS$. Let $\query$ and $\query'$ be $(s,\ws,\WC)$-queries over $\VS$ and $\TS$.
If $\query$ and $\query'$ are satisfiable, it holds, that:
$$ \query \contResp{\Gamma} \query' \iff \query'\Hom\query.$$  
\end{theorem}
\setcounter{theorem}{\value{restoreAppTheorem}}

\begin{proof} The Theorem~\ref{thm:queryContVsHom} is an immediate consequence 
	of Proposition \ref{thm:dswgquery-queryContVsHom} for 
	\dswgqueries and 
	of Proposition~\ref{thm:swggquery-queryContVsHom} in the case of 
	\swggqueries. (Recall that \dswgqueries are satisfiable in any way.) 
\end{proof}

\begin{Proposition}\label{thm:dswgquery-queryContVsHom}
	Let $\query$ and $\query'$ be $(\ell,\ws,\wc)$-\dswgqueries over $\VS$ and $\TS$.
	\begin{enumerate}
		\item\label{item:QuerContVsHom_HomToQC} 
		If $\query'\Hom\query$ then 
		$\query \contResp{\Gamma} \query'$.
		\item\label{item:QuerContVsHom_QCToHom}
		Let $\TSalt\subseteq\TS$ be such that
		$|\TSalt| \geq 2$ and $\TSalt\supseteq\types(\query)$.
		If $\query \contResp{\Delta} \query'$ 
		then $\query'\Hom\query$.
	\end{enumerate} 
\end{Proposition}

\begin{proof}
	For the proof
	 let $s$ and $s'$ be the query strings of $\query$ and $\query'$,
	respectively.
	\smallskip
	
	\noindent
	\eqref{item:QuerContVsHom_HomToQC}: 
	Let $h$ be a homomorphism from $\query'$ to $\query$. It is to show that $\query \contResp{\Gamma} \query'$. Let $t$ be arbitrary choosen from $\mods{\query}$.
	Our aim is to show that $\trace\in\mods{\query'}$.
	Therefore, we consider a witness $(\mu,e)$ for $\trace\models\query$. Recall, $e$ is an embedding $e : [\ell] \to [|t|]$ and $\mu$ is a substitution of size $\ell$, such that:
	\begin{enumerate}[1.)]
	\item For every position $i \in [\ell]$ where $s[i] \in \typesets(q)$, we have $t[e(i)] \in s[i]$. \footnote{Recall, we switch between elements of $\Gamma$ and singelton subsets of $\Gamma$ in the context of letters in a string.} 
	\item For every fixed variable $z \in \vars(q)$ and all positions $i_1, \ldots i_k \in [\ell]$, with $z=s[i_1]=\ldots=s[i_k]$, we have $\mu(i_1,s[i_1])= \ldots = \mu(i_k,s[i_k])= \mu(z) = t[e(i_1)] = \ldots =t[e(i_1)]$.
	\end{enumerate}
	Let $\mu'$ be defined via $\mu'(z)=\mu(h(z))$. We claim that $(\mu',e)$ is a witness for $\trace\models\query'$, 	i.e., $\mu'(s') \subseq_e t$.
	We already know that $(\mu,e)$ is a witness for $\trace\models\query$. 
	Hence, $e$ satisfies the global window size $\ws$ and the local gap-size constraints $\wc$, and for all $i\leq\ell$ we have $\mu(s[i])=\trace[e(i)]$.
	Consider an arbitrary $i\in[\ell]$. 
	We need to show that $\mu'(s'[i])=t[e(i)]$. 
	By our choice of $\mu'$ we have $\mu'(s'[i])=\mu(h(i,s'[i]))$.
	Since $h$ is a homomorphism, we have 
	\begin{enumerate}[1.)]
	\item for every variable $z'$ on position $i$ of $s'$ that is mapped via $h$ to a variable $z$ in $s$, that: \[ \mu'(z') = \mu'(s'[i]) = \mu(h(z'))= \mu(z) = t[e(i)] \]
	\item for every variable $z'$ that only occurs at one position $i$ in $s'$ and that is mapped via $h$ to a (non-empty) subset $\dc \finsubseteq \Gamma$, that: \[ \mu'(z') = \mu'(s'[i]) = \mu(i,h(z'))) = \mu(i, \dc) =  t[e(i)] \]
	\item for every variable $z'$ that occurs at least twice at positions $i_1, \ldots, i_k \in [\ell] $ in $s'$ and that therefore is mapped via $h$ to an element  $\gamma \in \Gamma$, that: 
	\begin{multline*}
 \mu'(z') = \mu'(s'[i_1]) = \ldots = \mu'(s'[i_k]) \\   =\mu(i_1,h(z'))) = \ldots = \mu(i_k,h(z')))  \\= \mu( \gamma) = \gamma =  t[e(i_1)]= \ldots=  t[e(i_k)] \end{multline*} and, finally
	\item for every (non-empty) $\dc' \finsubseteq \Gamma$  at position $i$ in $s'$ that is mapped via $h$ to an (non-empty) $\dc \subseteq \dc'$, that: \[ \mu'(i,\dc') = \mu'(s'[i]) = \mu(i,h(\dc'))) =  \mu(i, \dc) =  t[e(i)] \in \dc \subseteq \dc'  \] and therefore $t[e(i)] \in \dc'$. 
	\end{enumerate}
	In the end, it proves that $(\mu',e)$ is a witness for $\trace\in\mods{\query'}$.

	\smallskip
	
	\noindent
	\eqref{item:QuerContVsHom_QCToHom}:
	Let $\TSalt\subseteq\TS$ with $|\TSalt|\geq 2$ and $\TSalt\supseteq\types(\query)$, and let
	$\query \contResp{\TSalt} \query'$.
	Our aim is to show that $\query'\Hom\query$. We claim that $h(i,s'[i])=s[i]$ is an homomorphism from $\query'$ to \query.

	We fix an arbitrary $\type_0\in\TSalt$. For every $i<\ell$
	let $g_i$ be the
	``gap string'' consisting of $\wc^-_i$ copies of the symbol
	$\type_0$. Let
	\[
	\tilde{s} \ := \quad
	s[1]\ g_1 \ s[2]\ g_2 \ \cdots \ s[\ell{-}1]
	\ g_{\ell-1}\ s[\ell]\,.
	\]
	For each substitution $\mu : ([\ell] \times\VS \cup \Pfinplus(\TSalt)) \to \TS$  of size $\ell$ (recall, there are only singeltons in the 
co-domain of $\mu$ and every type of 
$\mu(s)$ is the unique element of its singleton) consider the trace $\trace_\mu\deff\mu(\tilde{s})$. 
	Obviously, we have $\trace_\mu\models \query$, as this is witnessed by $(\mu,e)$
	where $e(1)=1$, $e(2)=2+\wc^-_1$, \ldots,
	$e(j)=j+\sum_{i<j} \wc^-_i$ for all $j\in[\ell]$ (note, by assumption, we have $\TSalt\supseteq\types(\query)$). 
	
	Furthermore, we have $\query \contResp{\TSalt} \query'$, that implies $\trace_\mu\models \query'$. 
	Let $(\zeta_\mu,e_\mu)$ be a witness for $\trace_\mu\models\query'$. 
	Since $\trace_\mu$ has
	length exactly $\ell+\sum_{i<\ell} \wc^-_i$, there exists only one embedding that
	satisfies the local gap-size constraints $\wc$, namely the embedding $e$. I.e.,
	$e_\mu=e$. Furthermore, since $(\zeta_\mu,e)$ is a witness for
	$\trace_\mu\models\query'$, we know that
	$\zeta_\mu(s'[i])=\trace_\mu[e(i)]$ for all $i\leq\ell$.
	And by our choice of $\trace_\mu$ and $e$ we have $\trace_\mu[e(i)] =\mu(s[i])$ for
	all $i\leq \ell$. I.e.,
	\begin{equation}\label{eq:QueryContVsHom}
		\zeta_\mu(s'[i])
		\ = \ 
		\mu(s[i])
	\end{equation}
	for all  $i\leq \ell$ and all substitutions $\mu:([\ell] \times\VS \cup \Pfinplus(\TSalt)) \to \TS$.

	Now consider an arbitrary $i\in[\ell]$.\begin{enumerate}[\text{Case}~1:]
	\item  $s'[i]=\dc\finsubset\TS$.
	
	Then, by definition of $h$ we have $h(s'[i])=h(i,\dc)=s[i]$. 
	We have to show that $s[i]\subseteq s'[i]$.
	For contradiction, assume that $s[i]\not\subseteq s'[i]$, 
	hence there exists a type $\type\in s[i]\setminus s'[i]$.
	
	Then, let $\mu:([\ell] \times\VS \cup \Pfinplus(\TSalt)) \to \TS$ be a substitution with
	$\mu(s[i])=\type$.
	Then,  by \eqref{eq:QueryContVsHom}, we have
	$\mu(s[i]) = \zeta_\mu(s'[i]) = \type$.
	But due to the definition of substitution $\zeta_\mu(s'[i])\neq \type$, since $\type\not\in\dc$, contradicting our choice of $\mu$ and $\query \contResp{\TSalt} \query'$, respectively.
	
	\item $s'[i]\in\vars(\query')$. Let $x\deff s'[i]$. Again, by definition, we have $h(s'[i])=h(i,x)=s[i]$. We are done, iff $x$ occurs only once in $s'$. 	
	
	Otherwise, let $i_1, \ldots,  i_k$ be elements of $[\ell]$, such that $s'[i] =s'[i_1] = \ldots =s'[i_k] =  x $. 
	We have to show that \begin{quote}
	\begin{enumerate}[a)] 
		\item $s[i]=s[i_1] = \ldots = s[i_k] $. 
		\item if $s[i] \finsubseteq \Delta$ then $s[i]$ is a singelton. 
	\end{enumerate}\end{quote}
	
	\begin{description}
	\item[\textnormal{proof of a)}] For contradiction, assume that $s[j]\neq s[j']$ for some $j,j' \in \set{i,i_1,\ldots,i_k}$.
	Then, let $\mu:([\ell] \times\VS \cup \Pfinplus(\TSalt)) \to \TS$ be a substitution with 
	$\mu(s[j])\neq \mu(s[j'])$ 
	(such a substitution exists because $|\Delta|\geq 2$ and 
	$\Delta\supseteq\types(\query)$).
	Then, by \eqref{eq:QueryContVsHom}, we have
	$\zeta_\mu(s'[j])\neq \zeta_\mu(s'[j'])$,
	contradicting our choice of $\mu$ and $\query \contResp{\TSalt} \query'$, respectively.
	\item[\textnormal{proof of b)}] We already know, that $s[i]=s[i_1] = \ldots =  s[i_k] $. For contradiction, we assume that there are $\gamma_1 , \gamma_2 \in \Gamma$ such that  $\gamma_1 \neq \gamma_2$ and $\gamma_1 , \gamma_2 \in \dc \deff s[i_1] = s[i_2]$. Now again, let $\mu:([\ell] \times\VS \cup \Pfinplus(\TSalt)) \to \TS$ be a substitution with 
	$\gamma_1 = \mu(s[i_1])\neq \mu(s[i_2])= \gamma_2$. It holds that $\trace_\mu \models \query$, but for every $\trace'$ with $\trace'\models\query'$ and its witness $(\mu',e)$ we need $\mu'(i_1,x)=\mu'(i_2,x)$, so it contradicts $\query \contResp{\TSalt} \query'$. \qedhere
	\end{description}
	\end{enumerate}
\end{proof}

\begin{Proposition}\label{thm:swggquery-queryContVsHom}
	Let $\query$ and $\query'$ be $(\ell,\ws,\wC)$-\swggqueries over $\VS$ and $\TS$. Let $\query$ and $\query'$ be satisfiable.
	\begin{enumerate}
		\item\label{item:swggquery-QuerContVsHom_HomToQC} 
		If $\query'\Hom\query$ then $\types(\query')\subseteq \types(\query)$ 
		and $\query \contResp{\Gamma} \query'$  
		\item\label{item:swggquery-QuerContVsHom_QCToHom}
		Let $\TSalt\subseteq\TS$ be such that $|\TSalt| \geq |t|-\ell+|\types(q)|+1$ for $t \in \Mintrace{\TSalt}{\query}$ and 
		$\TSalt\supseteq\types(\query)$. Then the following is true.
		\[ \text{If $\query \contResp{\Delta} \query'$ then 
		$\query'\Hom\query$.} \]
	\end{enumerate}
\end{Proposition}

\noindent Before we start to prove Proposition~\ref{thm:swggquery-queryContVsHom}, let us consider the size of $\TSalt$ and thereby $\Gamma$. For every query $\query$ we defined \[\Mintrace{\TSalt}{\query} \deff \left\{ \trace\in\TSalt^+ \ :\ 
        \begin{array}{l}\trace\models\query \text{ and } \\
        \text{there exists no } \trace'\in\TSalt^+ \text{ with }  
            \trace'\models\query \text{ and } |\trace'|<|\trace| \end{array}
    \right\}.\]
\begin{Example} \label{ex:Prop12}
	If $|\Gamma|$ is big enough, then two query strings $s$ und $s'$ provide a
	good intuition, whether $q  \contResp{\Gamma} \query'$ holds, or not. 
	But if the size of $\Gamma$ is small and the queries string size is small 
	compared to the size of traces $t \in \Mintrace{\Gamma}{q}$ then 
	containment may hold since the essence of $\query'$ can be hidden in the 
	gaps of $s$, while there exists no homomorphism from $\query'$ to $\query$.
	
	Let $\Gamma=\set{\ta, \tb , \tc}$. Let $\query = (s,w,C)$ and 
	$\query' = (s',w,C)$ for $C=\set{(5,5,3)_1}$ be \swggqueries{} with  
		$s=\ta\tb\tc\ta $ 
	and 
		$s' = \ta x x \ta$. 
	Caused by the condition $(5,5,3)_1$, we have 
		$|\trace|=7$
	for every $\trace \in \Mintrace{\Gamma}{\query}$. 
	Hence, there exist $i$ and $j \in \set{2,\ldots,6}$, $i \neq j$, with 
		$\trace[i]=\trace[j]$,
	since we only have three symbols in $\TS$ to fill up the gap. 
	Therefore, for every $\trace \in \Mods{\TS}{q}$ it holds that 
		$\trace \in \Mods{\TS}{q'}$,
	but there exists no homomorphism $h$ from $\query'$ to $query$.
	To avoid this repetition of types within the gap and to ensure that 
		$\query \not\contResp{\Delta} \query'$, 
	three (instead of one) additional types are needed, apart from $\tb$ and 
	$\tc$. 

\end{Example}
\begin{proof}[of Proposition~\ref{thm:swggquery-queryContVsHom}]

	For the proof 
	 let $s$ and $s'$ be the query 
    strings	of $\query$ and $\query'$, respectively.
    \medskip

	\noindent
	\eqref{item:swggquery-QuerContVsHom_HomToQC}:
		The proof of \eqref{item:swggquery-QuerContVsHom_HomToQC} is exactly the proof of Proposition~\ref{thm:dswgquery-queryContVsHom}\eqref{item:QuerContVsHom_HomToQC} for the special case that the homomorphism maps directly to elemens of $\Gamma$.
    \medskip

	\noindent
	\eqref{item:swggquery-QuerContVsHom_QCToHom}:   
	Let  $\query$ and $\query'$ be satisfiable and let $\TSalt\subseteq\TS$ be such that $|\TSalt| \geq |t|-\ell+|\types(q)|+1$ for $t \in \Mintrace{\TSalt}{\query}$ and 
		$\TSalt\supseteq\types(\query)$.

	Our aim is to show that $\query'\Hom\query$. We claim that $h(s'[i])=s[i]$ is an homomorphism from $\query'$ to \query.
	
	We fix an arbitrary $\type_0\in\TSalt$ that does not occur in $s$.  
	Let $\TSalt_{\nS} \deff \TSalt \setminus \set{\types(q) \cup \set{\type_0}}$ 
	be the set of at least $|t|-\ell$ new types. We call the set  
		$\TSalt_{\base} \deff \TSalt  \setminus \TSalt_{\nS}$ 
	the set of base types. 
	Next, let $t_{\mathsf{min}}$ be an arbitrary element from (the non-empty set) $\Mintrace{\TSalt}{q}$. We consider the following string from $\TSalt^{+}$ of length $|t_{\mathsf{min}}|$:
	\[
	\tilde{\trace} \ := \quad
	s[1]\ g_1 \ s[2]\ g_2 \ \cdots \ s[\ell{-}1]
	\ g_{\ell-1}\ s[\ell]\,
	\] where $s$ is a subsequence of $\tilde{t}$ with embedding $e$ that 
	satisfies the set of generalised gap-size constraints $\wC$ and 
	$g_1, \ldots, g_{\ell-1}$ are strings over $\TSalt^+_\nS$ in a way, such 
	that any symbol $\type \in \TSalt^+_\nS$ occurs only once in 
	$g_1 \dots g_{\ell-1}$. 
	This is possible since for $t \in \Mintrace{\TSalt}{q}$ we have 
	\begin{align*}
		|\TSalt^+_\nS| &= |\TSalt| -( |\types(q)| +1)\\ 
		&\ge |t|-\ell+|\types(q)|+1  -( |\types(q)| +1) \\ 
		&=  |t|-\ell = \sum\limits_{i \in [\ell-1]} |g_i|.
	\end{align*} 
	Moreover, we choose the positions of $s[i]$ for all $i \in \set{2, \ldots, \ell -1}$ as left as possible with respect to the constraints $\wC$.

	Now, we choose an arbitary substitution 
		$\mu:\vars(s)\cup\TSalt_\base\to\TSalt_\base$.  
	Note, $\mu$ is an substitution since we have $\types(q) \subset \TSalt_\base$ and no letter of $s$ is mapped to a new symbol in $\TSalt_\nS$.
	
	Next, we consider the trace $\trace_\mu\deff\mu(\tilde{\trace})$. 
	Then $\trace_\mu\models \query$ holds, as this is witnessed by $(\mu,e)$. 
	By assumption, we have
        $\query \contResp{\Gamma} \query'$.  Hence, $\trace_\mu\models \query'$.
	Our choice of $t$ ensures that there is a $\zeta$ such that $(\zeta,e)$ witnesses
	$\trace_\mu\models\query'$.

	\begin{quote}
	Assume not. 
	Let $(\zeta',e')$ with $e'\neq e$ be a witness of $\trace_\mu\models \query'$. 
	This implies that at least one $s'[i]$ is mapped to a type $\type \in \TSalt_\nS$ via $\zeta'$, since it must be mapped into a gap string and can not be mapped to $\mu(s[i])$. Since $e$ satifies the conditions $\wC$, this can be caused by two reasons:
	\begin{itemize}
		\item 	$s'[i]=\type$. Then we obtain $\tilde{\trace}'$ from $\tilde{\trace}$ by replacing $\type$ by any arbitray $ \type'\neq\type  \in \TSalt$. 
			Again,  $\mu(\tilde{\trace}')$ contains no $\type$ and  $\mu(\tilde{\trace}') \models \query$. 
			That contradicts $\query \contResp{\TSalt} \query'$, as every trace $t$ with $t \models q'$ must have a position $j$ such that $t[j]=\type$.
		\item There exist an $x \in \vars(q')$ and some $j \in [\ell]$, such that  $s'[i]=s'[j]=x$. 
			First we remark that $s'[j]$ can not be mapped to some type $\type'$ from a gap string, cause all positions in gap strings of $\tilde{\trace}$ are pairwise disjoint by construction. 
			On the other hand, no $s'[j]$ can be mapped via $\zeta'$ to a position of $t_\mu$ obtained by some $\mu(s[i])$, since $\mu(s[i]) \in \TSalt_\base$. Both together implies $\zeta'(s'[i])\neq \zeta'(s'[i])$ but $s'[i]=s'[j]=x$, indicating that $\zeta'$ is not an substitution.
	\end{itemize}
	The remaing case that $s'[i]=x \in \vars(q')$ and $x$ occurs only once in the query string $s'$ can not prevent a witness with embedding $e$, hence we can choose $\zeta(x)= \zeta(s[i])$.
	\end{quote}
	\noindent
	The fact that $(\mu,e)$ witnesses $\trace_\mu \models \query$ and $(\zeta, e) $ witnesses $\trace_\mu \models \query'$ implies that 
	\begin{equation}\label{eq:QueryContVsHomElf}
	\mu(s[i]) = \zeta(s'[i]), \text{\quad \quad  for all $i\in[\ell]$.}
	\end{equation}
	Now consider an arbitrary $i\in[\ell]$ and recall that we want to prove that
	$h(s[i])=s'[i]$.
	
	\begin{enumerate}[\text{Case}~1:]
		\item $s'[i]\in\TSalt$. Let $\type\deff s'[i]$. Precisely, we have $\type\in\TSalt_\base$. 	
			Then, by definition of $h$ we have $h(s'[i])=h(\type)=\type$. 
			We have to show that $s[i]=\type$.
			For contradiction, assume that $s[i]\neq\type$. That implies:
			\begin{enumerate}[a)]
				\item If $s[i] = \type' \in \TSalt$ then $\mu(s[i]) = \mu(\type') = \type' \neq \type = \zeta(\type') = \zeta([s'[i]])$, contradicting~(\ref{eq:QueryContVsHomElf}).
				\item If $s[i] = x  \in \vars(q) $, then let $\Pos{s}{x}$ be the set of all positions $j \in [\ell]$ such that $s[j]=x$. We obtain $\tilde{\trace}'$ from $\tilde{\trace}$ by replacing $\type$ by any arbitray $ \type'\neq\type  \in \TSalt$ at all positions $e(j)$ for $j \in \Pos{s}{x}$. 
					Again, $\mu(\tilde{\trace}')$ does not contain $\type$ and  $\mu(\tilde{\trace}') \models \query$. That contradicts $\query \contResp{\TSalt} \query'$, as every trace $t$ with $t \models q'$ must include a position $j$ such that $t[j]=\type$.
			\end{enumerate}
	
		\item  $s'[i]\in\vars(\query')$. Let $x\deff s'[i]$. Again, by definition, we have $h(s'[i])=h(x)=s[i]$. We are done, if $x$ occurs only once in $s'$. 	
	
	Otherwise, let $i_1, \ldots, i_k$ be elements of $[\ell]$, such that $s'[i] =s'[i_1] = \ldots = s'[i_k] = x $. By definition of substitutions, we have $\zeta(s'[i])  = \zeta(s'[i_1]) = \ldots = \zeta(s'[i_k])$. We have to show that \[ s[i]=s[i_1] = \ldots= s[i_k] .\]

	 For contradiction, assume that $s[j]\neq s[j']$ for some $j,j' \in \set{i,i_1,\ldots,i_k}$.
		\begin{enumerate}[a)]
			\item Having $\TSalt \ni \type = s[j] \neq s[j'] = \type' \in \TSalt$ contradicts~(\ref{eq:QueryContVsHomElf}) since it implies 
			$\type = \mu(\type) = \mu(s[j]) \neq \zeta(s[j'])= \zeta(\type') = \type'$.
			\item Let $\vars(q) \ni y = s[j]  \neq s[j']  \in \TSalt \cup \vars(q)$.  Then let $\Pos{s}{y}$ be the set of all positions $i' \in [\ell]$ where $s[i']=y$.
				We obtain $\tilde{\trace}'$ from $\tilde{\trace}$ by replacing $\mu(s[j])$ by some arbitray $ \type'\neq \mu(s[j'])  \in \TSalt$  at all positions $e(i')$ for $i' \in \Pos{s}{y}$. 
				Again, $\mu(\tilde{\trace}')$ contains no $\type$ and  $\mu(\tilde{\trace}') \models \query$. That contradicts $\query \contResp{\TSalt} \query'$, as every trace $t$ with $t \models q'$ must have an position $j$ such that $t[j]=\type$.
		\end{enumerate}	 
	\end{enumerate}
	
	\noindent This completes the proof of Proposition~\ref{thm:swggquery-queryContVsHom}.

	\noindent 

	As a remark, we state that the choice of the size 
		$|\TSalt| \deff \max(1,\types(\query)) + |\trace| - \ell$
	for $\trace\in\min\Mintrace{\TSalt}{\query}$ is worst case minimal,
	therefore consider Example~\ref{ex:Prop12}. 
	\end{proof}

\section{About Isomorphisms}\label{app:iso}

\nc{\dcInPotTS}{\ensuremath{\dc \in \Pfinplus(\TS)}}
\nc{\InPotTS}[1]{\ensuremath{#1\in \Pfinplus(\TS)}}

\begin{Definition}\label{def:isomporphism}
    Two $(\ell,\ws,\WC,k)$-queries $\query=(s,\ws,\WC)$ and 
    $\query'=(s',\ws,\WC)$ are called \emph{isomorphic} (denoted by 
        $\query \cong\query'$)
    if there is a bijection 
        $\pi:(\vars(\query)\cup\Pfinplus(\TS))
         \rightarrow
         (\vars(\query')\cup\Pfinplus(\TS))$
    such that 
        $\pi(s[i]) = s'[i]$
    for all $i\in[\ell]$ and  $\pi_{|_{\Pfinplus(\TS)}} = \textnormal{id}$.
\end{Definition}

\begin{Corollary}\label{cor:Iso-Hom-Mod}
    Given some sufficiently large $\TS$. For all $(\ell,\ws,\WC,k)$-queries $\query$ and $\query'$
    over $\TS$ and $\Vars$ we have: 
    $$ 
        \query\cong\query' \Longleftrightarrow 
        \bigl( \query\Hom\query' \text{ and } \query'\Hom\query \bigr)
        \Longleftrightarrow 
        \Mods{\TS}{\query} = \Mods{\TS}{\query'}
    $$
\end{Corollary}

\begin{proof}
    Since we assume $|\TS|\geq 2$ in the case of \dswgqueries, or $|\TS|\geq |t|-|s|+|\types(q)|+1$ in the case of \swggqueries, respectivly, 
    the equivalence 
    $\bigl( \query\Hom\query' \text{ and } \query'\Hom\query \bigr)
    \Longleftrightarrow 
    \Mods{\TS}{\query} = \Mods{\TS}{\query'}$
    is a direct consequence of Theorem~\ref{thm:queryContVsHom}.
    \par 
    If $\query\cong\query'$, then there is a bijection 
    $\pi:(\vars(\query)\cup\Pfinplus(\TS))
         \rightarrow
         (\vars(\query')\cup\Pfinplus(\TS))$
    such that 
        $\pi(\dc) = \dc$ for all $\dcInPotTS$
    and 
        $\pi(s[i]) = s'[i]$
    for all $i\in[\ell]$. By definition, $\pi$ is also a homomorphism from 
    $\query$ to $\query'$, and $\pi^{-1}$ is a homomorphism from $\query'$ to
    $\query$ (note that since $\pi(\dc)=\dc$ for all $\dcInPotTS$ 
    and $\pi^{-1}$ is injective, $\pi^{-1}$ is also a substitution). Hence, 
    $\query\Hom\query'$ and $\query'\Hom\query$.
    \par 
    It remains to prove that 
        $\bigl( \query\Hom\query' \text{ and } \query'\Hom\query \bigr)$
    implies 
        $\query\cong\query'$.
    Let 
        $h: (\Vars \cup \Pfinplus(\TS))\to (\Vars\cup\Pfinplus(\TS))$
    and 
        $h': (\Vars \cup \Pfinplus(\TS))\to (\Vars\cup\Pfinplus(\TS))$
    be homomorphisms from $\query$ to $\query'$ and from $\query'$ to $\query$,
    respectively. (Note that we can omit the index $\ell$ and the parameter $i$
    for the position becasue it is given by the context.)
    By definition, this means that 
        $h(\dc) = h'(\dc) = \dc$ 
    for all $\dcInPotTS$, and 
        $h(s[i]) = s'[i]$ and $h'(s'[i]) = s[i]$
    for all $i\in[\ell]$.
    We claim that $h$ actually witnesses $\query\cong\query'$, i.e. $h$ is a 
    bijection
        $(\vars(\query)\cup\Pfinplus(\TS))
        \rightarrow
        (\vars(\query')\cup\Pfinplus(\TS))$
    such that 
        $h(\dc) = \dc$
    for all $\dcInPotTS$, and 
        $h(s[i]) = s'[i]$
    for all $i\in[\ell]$.
    \par 
    We already observed that 
        $h(\dc) = \dc$ 
    for all $\dcInPotTS$, and
        $h(s[i]) = s'[i]$ 
    for all $i\in[\ell]$ holds. Thus, it only remains to prove that $h$ is a 
    bijection.
    \par 
    Let $x,y\in\vars(\query)\cup\Pfinplus(\TS)$ with 
    $x\neq y$ and $h(x) = h(y)$. 
    If $\InPotTS{x,y}$, then
        $h(x) = x \neq y = h(y)$,
    which contradicts $h(x) = h(y)$.
    If $x\in\vars(\query)$ and $y=\dcInPotTS$,  
    then there exists a $p\in[\ell]$ with $s[p]=x$ and, since $h(s[p]) = s'[p]$ 
    and
        $h(s[p]) = h(x) = h(y) = \dc$,
    we have $s'[p] = \dc$, which contradicts $h'(s'[p]) = s[p] = x$.
    The case where $\InPotTS{x}$ 
    and $y\in\vars(\query)$ can be dealt with analogously.
    If $x,y\in\vars(\query)$, then there are $p,r\in[\ell]$ with 
        $s[p] = x$ and $s[r] = y$.
    Since $h(s[p]) = s'[p]$, $h(s[r]) = s'[r]$ and $h(s[p])=h(s[r])$ by our 
    assumption it holds that $s'[p] = s'[r]$.
    However, this implies that 
        $h'(s'[p]) = h'(s'[r])$,
    i.e. $s[p] = s[r]$, which contradicts the assumption that $x\neq y$.
    Consequently, $h$ is injective.
    \par 
    In order to prove that $h$ is surjective, let 
        $x\in\vars(\query')\cup\Pfinplus(\TS)$.
    If $x = \dcInPotTS$ 
    then $h(\dc)=\dc$ by the definition of $h$ and since $h'$ exists as well.
    If $x\in\vars(\query')$, then there exists a position $p\in[\ell]$ with 
        $s'[p] = x$.
    Since $h$ is a homomorphism from $\query$ to $\query'$, it satisfies    
        $h(s[p]) = s'[p] = x$,
    which means that there exists 
        $y\in \vars(\query)\cup\Pfinplus(\TS)$ 
    with $h(y) = x$ and $y=s[p]$.
    Consequently, $h$ is surjective. Finally, we have shown that $h$ is injective and surjective, and therefore $h$ is a bijection. 
\end{proof}

\begin{Definition}\label{def:partIsom}
    Let $\query=(s,\ws,\WC)$ and $\query'=(s',\ws,\WC)$ be two 
    $(\ell,\ws,\WC,k)$-queries and $I\subset[\ell]$. We say that 
    $\query$ is partially isomorphic to $\query'$ w.r.t $I$
    (denoted by $\query \sim_{I}\query'$) if, and only if 
    \begin{enumerate}
        \item\label{def:partIsom-i} for all $i,j\in I$ we have:
            \[ s[i]=s[j] \Leftrightarrow s'[i]=s'[j] , \qquad \text{ and}\]
        \item\label{def:partIsom-ii}  for all $i\in I$ we have:
            \[
                \InPotTS{s[i]}
                \Leftrightarrow 
                \InPotTS{s'[i]}
        			\qquad \text{ and  } \qquad 
              \InPotTS{s[i]} \Rightarrow s[i] = s'[i]. \]
    \end{enumerate}
\end{Definition}

\begin{Lemma}\label{lem:SimToIsom}
    For all $(\ell,\ws,\WC,k)$-queries $\query$ and $\query'$ we have:
    $$\query \sim_{[\ell]}\query'\ \Longleftrightarrow\ \query\cong\query'.$$
\end{Lemma}

\begin{proof}
    Assume that $\query\cong\query'$. Then there exists a bijection 
        $\pi:(\vars(\query)\cup\Pfinplus(\TS))
        \rightarrow
        (\vars(\query')\cup\Pfinplus(\TS))$
    such that 
        $\pi(\dc) = \dc$ for all $\dcInPotTS$ 
    and 
        $\pi(s[i]) = s'[i]$
    for all $i\in[\ell]$.
    Since $\pi$ is the identity on $\Pfinplus(\TS)$, it 
    holds for all $i\in[\ell]$ with $\InPotTS{s[i]}$
    that $s[i] = \pi(s[i]) = \InPotTS{s'[i]}$.
    This also implies $s'[i]=s'[j]$ for all $i,j\in I$ with 
        $s[i]=s[j]$
    and 
        $\InPotTS{s[i],s[j]}$.
    For all $i,j\in I$ with $s[i]=s[j]$ and $s[i],s[j]\in\vars(\query)$ it holds 
    that 
        $s'[i],s'[j]\in\vars(\query')$.
    Since $\pi$ is injective, $\pi(s'[i]) = \pi(s'[j])$ holds, which in turn 
    implies $s'[i] = s'[j]$.
    \par
    For direction "$\Longrightarrow$" we have $\query\sim_{[\ell]}\query'$ by 
    assumption. Let $s$ and $s'$ be the query strings of $\query$ and $\query'$,
    respectively. For every $x\in\vars(\query)$ let 
        $i_x\deff \min\{ i\in[\ell] \ :\  s[i]=x\}$.
    Define 
        $\pi:
        (\vars(\query)\cup\Pfinplus(\TS))
        \rightarrow
        (\vars(\query')\cup\Pfinplus(\TS))$
    via 
        $\pi(\dc)=\dc$ for all $\dcInPotTS$ 
    and 
        $\pi(x) = s'[i_x]$ for all $x\in\vars(\query)$.
    \par 
    First, note that $\pi$ is injective: Consider $x,y\in\vars(\query)$ with 
    $\pi(x) = \pi(y)$. Then, $s'[i_x] = s'[i_y]$. By item~\ref{def:partIsom-i}
    we obtain $s[i_x]=s[i_y]$, i.e. $x=y$.
    \par 
    Furhtermore, $\pi$ is surjective: Consider an arbitrary 
    $y\in\vars(\query')$. Let 
        $j_y\deff \min\{ j\in[\ell] \ :\  s'[i]=y\}$.
    Let $x\deff s[j_y]$. Then, $x = s[i_x] = s[j_y]$, and by 
    item~\ref{def:partIsom-i} of Definition~\ref{def:partIsom} we obtain that 
    $s'[i_x] = s'[j_y]$. Hence, $y = s'[j_y] = s'[i_x] = \pi(x)$.
    \par 
    In summary, $\pi$ is a bijection from 
    $(\vars(\query)\cup\Pfinplus(\TS))$ to 
    $(\vars(\query')\cup\Pfinplus(\TS))$ with 
    $\pi(\dc)=\dc$ for all $\dcInPotTS$.  
    It remains to prove that $\pi(s[i]) = s[i]$ for all $i\in[\ell]$.
    Consider an arbitrary $i\in[\ell]$. 
    If $\InPotTS{s[i]}$ then, by 
    item~\ref{def:partIsom-ii} of Definition~\ref{def:partIsom} and the 
    definition of $\pi$, we have 
        $\pi(s[i]) = s[i] = s'[i]$.
    Assume $s[i] = x \in\vars(\query)$. Since $s[i] = s[i_x]$, we obtain from 
    item~\ref{def:partIsom-i} of Definition~\ref{def:partIsom} that 
        $s'[i] = s'[i_x]$.
    Thus, $\pi(s[i]) = \pi(x) = s'[i_x] = s'[i]$, which completes the proof of 
    Lemma~\ref{lem:SimToIsom}. 
\end{proof}

\begin{Remark}\label{rem:DescriptiveIsHomDescriptive}
   Given some sufficiently large $\TS$. Hence, by Theorem~\ref{thm:queryContVsHom} 
    and Corollary~\ref{cor:Iso-Hom-Mod} we know that a query $\query$ is 
    descriptive for a sample $\Sample$ w.r.t $(\supp,(\ell,\ws,\WC))$ if, 
    and only if, $\query$ is an $(\ell,\ws,\WC)$-query with 
        $\Supp(\query,\Sample)\geq \supp$, 
    and there is no other $(\ell,\ws,\WC)$-query $\query'$ with 
        $\Supp(\query,\Sample)\geq \supp$
    and $\query\Hom\query'$ and $\query\not\cong\query'$.
\end{Remark}

\section{About Discovering}\label{app:discover}

We briefly discuss why $\TSalt$ is passed through incrementally in 
Line~\ref{algo:forDeltaiLoopLine} on an intuitive level by picking up the 
example in section~\ref{sec:discovery-short}. 
\setcounter{restoreAppTheorem}{\value{theorem}}
\setcounter{theorem}{\value{Counter_Exa:Run}}
\begin{Example}[extended]
    Let $\TS=\{\ta,\tb,\tc\}$ and $x_1, x_2, x_3\in\VS$. Consider the sample    
    $\Sample=\{\ta\tb\tb,\ta\tc\tc\}$, $\supp=1.0$, $\ell=\ws=3$, 
    $\wc = ((0,0),(0,0))$ and $k=2$.
    \par 
    On input $(\Sample,\supp,(\ell,\ws,\wc,2))$ the algorithm first 
    generates $\query = (x_1 x_2 x_3,\ws,\wc)$. Since $\query$
    satisfies the support threshold the algorithm proceeds by computing 
        $\TSalt =   \{ \{\ta\}\} \dot\cup 
                    \{\{\ta,\tb\}, \{\ta,\tc\}, \{\tb,\tc\} \}$.
    Assume the algorithm selects $x:=x_3$ during the first iteration of the 
    main loop. It turns out that $\TSalt_1 = \{ \{\ta\}\}$ does not contain a 
    typeset for an admissable replacement of $x_3$. Hence, the algorithm 
    considers $\TSalt_2$ during the second transition of the for-loop in 
    Line~\ref{algo:forDeltaiLoopLine}. The only admissable replacement is 
        $\replace{s}{x_3}{\{\tb,\tc\}}$, 
    and $s$ is replaced by $x_1 x_2 \{\tb,\tc\}$ ($\AV$ remains empty).
    \par 
    Let us assume that during the second transition through the main loop the
    algorithm selects $x:=x_1$ and $y:=\{\ta\}\in\TSalt_1$.
    The replacement of $x_1$ by $\{\ta\}$ is admissible 
    (as it has support 1 on $\Sample$). Therefore, $s$ is replaced by 
        $\{\ta\} x_2 \{\tb,\tc\}$
    and $\AV$ remains unchanged again.
    \par
    In its last iteration (during the second transition through the for-loop),
    $\replace{s}{x_2}{\{\tb,\tc\}}$ is the only admissible replacement 
    operation. The run terminates after this iteration and outputs the query 
    $\query=(s,\ws,\wc)$ with $s=\{\ta\} \{\tb,\tc\} \{\tb,\tc\}$.
    \par \vspace*{0.5em}
    Now consider the case where Line~\ref{algo:forDeltaiLoopLine} is
    omitted and $y$ is chosen from $\TSalt\cup\AV$. Then
    replacing $x_1$ by $\{\ta,\tb\}$ is an admissable replacement, but the 
    resulting query $\query'=(s',\ws,\wc)$ with 
        $s'= \{\ta,\tb\} \{\tb,\tc\} \{\tb,\tc\}$ 
    is not be descriptive (due to 

    $\query\contResp{\TS}\query'$).
\end{Example}
\setcounter{theorem}{\value{restoreAppTheorem}}

\setcounter{theorem}{\value{Counter_Thm:Thm:discovery}}
\begin{theorem}\textbf{(restated)}
    Given some sufficiently large $\TS$.
    Let $\Sample$ be a sample, let $\supp$ be a support threshold with 
    $0<\supp\leq 1$, let $(\ell,\ws,\WC,k)$ be query parameters with $k=1$ 
    if $\WC=\wC$.
	\begin{enumerate}[(a)]
		\item\label{item:ComputeDescrQuery:error}
            If there does not exist any $(\ell,\ws,\WC,k)$-swgg-- or dswg--query 
            that is descriptive for $\Sample$ w.r.t.\ $(\supp,(\ell,\ws,\WC,k))$
            then there is only one run of
            Algorithm~\ref{algo:ComputeDescQueryFromq} upon the defined input,
            and it stops in line~\ref{algo:ErrorMsgLineBeginning} with output 
            $\Error$.
		\item\label{item:ComputeDescrQuery:EveryOutputIsDescr}
		    Otherwise, every run of Algorithm~\ref{algo:ComputeDescQueryFromq} 
            upon input $(\Sample,\supp,(\ell,\ws,\WC,k))$ terminates and 
            outputs an \swggquery{} or \dswgquery{} $\query$ (depending on $k$), 
            with $|\dc|\leq k$ for all $\dc\in\typesets(\query)$, that is 
            descriptive for $\Sample$ w.r.t.\ $(\supp,(\ell,\ws,\WC,k))$.
	\end{enumerate}
    \vspace*{-1em}
\end{theorem}
\setcounter{theorem}{\value{restoreAppTheorem}}

\begin{proof}
    First, consider the case that there does not exist any swgg-- or dswg--query
    with parameters $(\ell,\ws,\WC,k)$ that is descriptive for $\Sample$ w.r.t.\ 
    $(\supp,(\ell,\ws,\WC,k))$.
    Note that this implies $\Supp(\mgq,\Sample)<\supp$, whereby $\mgq$ is the 
    most general query for $(\ell,\ws,\WC,k)$, 
    Recall that the query string of $\mgq$ is of form 
        $\mgs = x_1 \dots x_\ell$, 
    and $\mgq$ is most general in the sense that 
        $\query' \contResp{\TS} \mgq$
    for each $(\ell,\ws,\WC,k)$-query $\query'$. Hence, for every 
    $(\ell,\ws,\WC,k)$-query $\query'$ with $\mods{\query}\subseteq\mods{\mgq}$
    it holds that $\Supp(\query,\Sample)\leq \Supp(\mgq,\Sample)<\supp$.
	Therefore, the algorithm stops in Line~\ref{algo:ErrorMsgLineBeginning} and 
    outputs an error message $\Error$.
	This proves statement \eqref{item:ComputeDescrQuery:error}.
    
    The second statement \eqref{item:ComputeDescrQuery:EveryOutputIsDescr} of
    Theorem~\ref{thm:ComputeDescrQuery} is an immediate consequence of the
    Propositions~\ref{prop:ComputeDescQueryFromq-swgg} and
    \ref{prop:ComputeDescQueryFromq-dswg} for \swggqueries{} and \dswgqueries,
    respectively. 
\end{proof}

\begin{Proposition}\label{prop:ComputeDescQueryFromq-swgg}
    Given some sufficiently large $\TS$.
    Let $\Sample$ be a sample, let $\supp$ be a support threshold with 
    $0<\supp\leq 1$, let $(\ell,\ws,\wC)$ be query parameters and let 
    $\mgq$ be the most general query for $(\ell,\ws,\wC)$. 
    In case that $\Supp(\mgq,\Sample)\geq\supp$, every run of 
    Algorithm~\ref{algo:ComputeDescQueryFromq} upon input 
        $(\Sample,\supp,(\ell,\ws,\wC))$ 
    terminates and outputs an \swggquery{} $\query$, that is descriptive for 
    $\Sample$ w.r.t.\ $(\supp,(\ell,\ws,\wC))$.
\end{Proposition}

\begin{proof}[Proof (Sketch).]
    For the special case where $\mgq$ is an \swgquery, \cite{KSSSW22} provided 
    the following result: 
    The algorithm obtained from Algorithm~\ref{algo:ComputeDescQueryFromq} by 
    starting with an arbitrary input query $\query$ instead of $(\ell,\ws,\WC)$,
    outputs either a query $\query'$ that is descriptive for $\Sample$ w.r.t.\ 
    $(\supp,(\ell,\ws,\wc))$ and satisfies 
        $\query' \contResp{\TS} \query$
    or, in case that no such $\query'$ exists, the message $\Error$.
    \par 
    Note that despite the slight difference regarding the input parameters the 
    algorithms only differ in the matching routine due to the generalised
    gap-size constraints:
    because of $\Supp(\mgq,\Sample)\geq\supp$, for the most general query 
    defined in line~\ref{algo:MgqDefLine}, every run of 
    $\QComputeDescrQueryFromq$ will reach 
    Line~\ref{algo:mainLoopLine} and proceed from there on.
    Since the query string of an \swggqueries{} is solely defined over 
    $\Vars\cup\TS$, i.e. $k=1$, $\TSalt$ is set to 
    $$\TSalt_1 = \{\type\in\TS:
         \frac{|\{\trace\in\Sample\ :\ \type\in\types(\trace)\}|}{|\Sample|}
         \geq\supp\}
    $$
    in line~\ref{algo:TSaltDefLine} and is extended by the set of currently 
    available variables $\AV$ in line~\ref{algo:selectNextPositionLine} for 
    each iteration of the main loop. 
    Note that the for-loop is only passed once during each transition through 
    the main loop.
    \par 
    Hence, a run of Algorithm~\ref{algo:ComputeDescQueryFromq} for $k=1$ equals 
    a run of the algorithm presented in \cite{KSSSW22} which gets as input 
    the Sample $\Sample$, the support threshold $\supp$ and the query $\mgq$,     
    except for a different black box matching routine.
    Thus, the results regarding the descriptiveness of the output query carry 
    over from \cite{KSSSW22} to the case that 
    Algorithm~\ref{algo:ComputeDescQueryFromq} computes an \swggquery. 
\end{proof}

\begin{Proposition}\label{prop:ComputeDescQueryFromq-dswg}

    Let $|\Gamma|\geq 2$. Let $\Sample$ be a sample, let $\supp$ be a support threshold with 
    $0<\supp\leq 1$, let $(\ell,\ws,\wc,k)$ be query parameters and let 
    $\mgq$ be the most general query for $(\ell,\ws,\wc,k)$. 
    In case that $\Supp(\mgq,\Sample)\geq\supp$, every run of 
    Algorithm~\ref{algo:ComputeDescQueryFromq} upon input 
        $(\Sample,\supp,(\ell,\ws,\wc,k))$ 
    terminates and outputs an \dswgquery{} $\query$, with 
        $|\dc|\leq k$ 
    for all $\dc\in\typesets(\query)$, that is descriptive for 
    $\Sample$ w.r.t.\ $(\supp,(\ell,\ws,\wc, k))$.
\end{Proposition}

\begin{proof}
    Throughout the proof, we heavily make use of
	Remark~\ref{rem:DescriptiveIsHomDescriptive}, that holds for \dswgqueries{}
    already if $|\Gamma|\geq 2$.

    Assume $\Supp(\mgq,\Sample)\geq\supp$, for the most general disjunctive 
    query defined in line~\ref{algo:MgqDefLine}.
	In this case, every run of $\QComputeDescrQueryFromq$
	will reach Line~\ref{algo:mainLoopLine} and proceed from there on.
	Let $\TSalt$ be the set of typesets defined in Line~\ref{algo:TSaltDefLine}.
	Note that every query $\query'$ that is descriptive for $\Sample$ w.r.t.\ 
    $(\supp,(\ell,\ws,\wc,k))$, satisfies 
        $\typesets(\query')\subseteq\TSalt$, 
    because otherwise, $\Supp(\query',\Sample)$ would be $<\supp$.
    Especially, for each query $\query$ that is computed by 
    Algorithm~\ref{algo:ComputeDescQueryFromq} it holds that 
        $\typesets(\query)\subseteq\TSalt$,
    and 
        $|\dc|\leq k$ 
    for all $\dc\in\TSalt$ (due to line~\ref{algo:TSaltDefLine}).
    Hence, all typesets occuring in $\query$ have size less or equal to $k$.
    \par
    Let us first argue that every iteration $r$ of the outer \textbf{while}-loop
    starting in Line~\ref{algo:mainLoopLine}, will eventually end.
	To see this, first, note that the set of available types and variables 
    $\TSalt\cup\AV$ will always be finite, since it is bounded by the number of 
    types and typesets occurring in the given sample and the number of variables
    in $\mgs$, which equals $\ell$.
    During the $i$-th (of a bounded number of $k$) iterations through the 
    for-loop, the inner while loop starting in Line~\ref{algo:innerLoopLine} 
    ends after at most $|\TSalt_i|$ iterations and during each iteration the 
    current variable is either replaced by a typeset or variable or 
	remains in the query string if no replacement operation is possible.
    \par 
    Let us now fix a particular run of $\QComputeDescrQueryFromq$.
    Let $\query_0=\mgq$, $s_0=\mgs$, $\NR_0\deff \vars(\query_0)$ and 
    $\AV_0\deff\emptyset$. 
	And for every $r\in\set{1,2,\ldots}$ let $s_r$, $\NR_r$, $\AV_r$ be the
	query string $s$ and the sets $\NR$ and $\AV$ at the end of the $r$-th 
    iteration through the \textbf{while}-loop, and let $q_r$ be the query 
    $(s_r,\ws,\WC)$.
	Furthermore, for each $r\geq 1$ let $x_r$ be the particular element in
	$\NR$ that is chosen at the beginning of the $r$-th iteration through
	the outer \textbf{while}-loop.
    \par 
    By induction on $r$ and by construction of the algorithm, it is
	straightforward to prove the following claim.

	\begin{Claim}\label{claim:algo-invariants}
		For every $r\geq 1$ we have
		\begin{enumerate}
			\item\label{item:algo-invariants:NRandAV} 
			$\NR_{r}=\vars(\query_0)\setminus\set{x_1,\ldots,x_r}$ and
			$\NR_r\cap\AV_r=\emptyset$ and
			$\NR_r\cup\AV_r=\vars(s_r)$.
			\item\label{item:algo-invariants:Supp} 
			$\typesets(\query_r)\subseteq\TSalt$ and $\Supp(\query_r,\Sample)\geq\supp$.
			\item\label{item:algo-invariants:Hom}
			$\query_{r-1}\Hom\query_r$. 
			\item\label{item:algo-invariants:NR-nonrepeated} 
			For every $x\in\NR_r$ we have $\Pos{s_r}{x}=\Pos{s_0}{x}$.
			\item\label{item:algo-invariants:PositionsNotInNRremainUnchanged}
			$s_r[j]=s_{r-1}[j]$ for all $j\in[\ell]\setminus\Pos{s_{0}}{x_r}$.
		\end{enumerate}
	\end{Claim}

    From this claim we obtain that after $\hat{r}:=|\vars(\query_0)|=\ell$
	iterations through the algorithm's outer \textbf{while}-loop, the 
    algorithm's run terminates with
	    $\NR_{\hat{r}}=\emptyset$ 
    and outputs an $(\ell,\ws,\wc,k)$-query $\query_{\hat{r}}$ with 
	    $\Supp(\query_{\hat{r}},\Sample)\geq \supp$ 
    and
	    $\query_0\Hom\query_{\hat{r}}$ 
    (i.e., by Theorem~\ref{thm:queryContVsHom},
        $\query_{\hat{r}}\contResp{\TS}\query_0$).
    \par 
    We need to show that this query $\query_{\hat{r}}$ is descriptive for 
	$\Sample$ w.r.t.\ $(\supp,(\ell,\ws,\wc,k))$. 
	\emph{For contradiction}, assume that it is not.
	Then, according to Remark~\ref{rem:DescriptiveIsHomDescriptive},
	there exists an $(\ell,\ws,\wc,k)$-query $\query'$ with
	    $\Supp(\query',\Sample)\geq\supp$
    and 
	    $\query_{\hat{r}}\Hom\query'$ 
    and
	    $\query'\not\cong\query_{\hat{r}}$.
    \par 
    From Claim~\ref{claim:algo-invariants}\eqref{item:algo-invariants:Hom} we
	know that 
        $\query_r\Hom\query_{\hat{r}}$ 
    for all $r\leq\hat{r}$, and hence
	    $\query_{\hat{r}}\Hom\query'$ 
    yields that
	\begin{equation}\tag{$\dagger$}\label{eq:HomToqStrich}
		\query_r\Hom\query' \quad \text{for all $r\in\set{0,1,\ldots,\hat{r}}$.}
	\end{equation}
	Let $s'$ be the query string of $\query'$.
	In order to deduce the desired contradiction, the notion of partially
	isomorphic queries 
	will be crucial.
	For each $r\in\set{0,1,\ldots,\hat{r}}$ let 
	\[
		I_r \ \deff\ \setc{i\in[\ell]}{s_0[i]\in\TS}\ \cup\
		\bigcup_{\nu=1}^r\Pos{s_0}{x_\nu}\,.
	\]
    Note that $I_0=\varnothing$ and $|I_r|=r$ for all $r\geq 1$, since 
    $\types(\query_0)=\varnothing$ and each variable occurs only once in $s_0$.
	The next claim provides the most crucial technical contribution of our
	proof.
    \begin{Claim}\label{claim:SimInduction}
		For every $r\in\set{0,1,\ldots,\hat{r}}$ we have
		$\query_r\sim_{I_r}\query'$.
	\end{Claim}
	Before turning to the proof of Claim~\ref{claim:SimInduction}
	let us first argue how the claim serves for completing the proof of 
	Theorem~\ref{thm:ComputeDescrQuery}.
	For $r=\hat{r}$ we know that $\NR_{\hat{r}}=\emptyset$.
	Hence, $\set{x_1,\ldots,x_{\hat{r}}}=\vars(\query_0)$ and 
    $I_{\hat{r}}=[\ell]$. 
    From Claim~\ref{claim:SimInduction}	we obtain 
        $\query_{\hat{r}}\sim_{[\ell]}\query'$. 
    But, by Lemma~\ref{lem:SimToIsom} this implies that
        $\query_{\hat{r}}\cong\query'$, 
    contradicting our assumption that
	    $\query_{\hat{r}}\not\cong\query'$. 
	Thus, all that remains to complete the proof of 
    Theorem~\ref{thm:ComputeDescrQuery} is to prove 
    Claim~\ref{claim:SimInduction}.
    \begin{proof}[Proof of Claim~\ref{claim:SimInduction}]\ \\
        We proceed by induction on $r$. 
        For the induction base with $r=0$ recall that $\query_0=\mgq$. Thus,
        $I_0 = \varnothing$, which immediately implies that 
            $\query_0 \sim_{I_0}\query'$.
        \par 
        For the induction step consider an arbitrary $r\geq 1$. At the beginning
        of the $r$-th iteration of the main loop the situation is as follows:
        $\NR = \NR_{r-1} \neq\varnothing$ and $s=s_{r-1}$ and, by 
        Claim~\ref{claim:algo-invariants}\eqref{item:algo-invariants:NRandAV},
            $\NR_{r-1}= \vars(\query_0)\setminus
                        \{ x_\nu\ :\ 1\leq \nu\leq r{-}1 \}$.
        Recall that by $x_r$ we denote the particular element of $\NR_{r-1}$ 
        chosen at the beginning of the $r$-th iterations through the main loop.
        \par 
        The induction hypothesis states that 
            $\query_{r-1} \sim_{I_{r-1}} \query'$
        holds. We have to show that 
            $\query_{r} \sim_{I_{r}} \query'$
        holds as well, whereby 
            $I_r = I_{r-1} \cup \Pos{s_0}{x_r}$.
        From \eqref{eq:HomToqStrich} we know that $\query_r \Hom \query'$. Thus,
        the following is true:
        \begin{enumerate}
			\item\label{item:sim-known-i}
			For all $i,j\in[\ell]$: \ If $s_r[j]=s_r[i]$ then $s'[j]=s'[i]$.\\
			\item\label{item:sim-known-ii}
			For all $i\in[\ell]$: \ 
            If $\InPotTS{s_r[i]}$
            then $\varnothing \subset s'[i] \subseteq s_r[i]$.
		\end{enumerate}
        Furthermore, by Claim~\ref{claim:algo-invariants}\eqref{item:algo-invariants:PositionsNotInNRremainUnchanged},
        $s_r$ coincides with $s_{r-1}$ on all positions $i\in[\ell]$ with 
            $i\not\in\Pos{s_0}{x_r}$.
        \par 
        Recall that $|\Pos{s_0}{x_r}|$=1 due to $s_0=\mgs$. To ease notation
        we simply write $p_r$ to denote the unique position of $x_r$ in $s_0$.
        Since
            $p_r\not\in I_{r-1}$,
        the induction hypothesis 
            $\query_{r-1} \sim_{I_{r-1}} \query'$ 
        hence implies that 
            $\query_{r} \sim_{I_{r-1}} \query'$.
        In order to prove that 
            $\query_{r} \sim_{I_{r}} \query'$, 
        it only remains to prove the following:
        \begin{enumerate}[i]
			\item\label{item:sim-todo-i} For all $j\in I_{r-1}$:
            \ If $s'[j]=s'[p_r]$ then $s_r[j]=s_r[p_r]$.\\
			\item\label{item:sim-todo-ii}
            If $\InPotTS{s'[p_r]}$ then
            $s_r[p_r]\subseteq s'[p_r]$.
		\end{enumerate}
        By the definition of the algorithm, the query $\query_r$ is obtained from 
        $\query_{r-1}$ by performing exactly one replacement operation using a 
        typeset (\TypeRepl) (line~\ref{algo:ReplaceOpLine} with 
        $y\in\TSalt_i\setminus\AV$), one replacement operation using a variable 
        (\VarRepl) (line~\ref{algo:ReplaceOpLine} with $y\in\AV$) or no replacement 
        operation (\NoChange) in line~\ref{algo:noChangeOpLine}. 
        Note that a \NoChange{} will only be performed if the following 
        conditions are true:
        \begin{equation}\label{eq:algo-sim-vars}\tag*{$(*)_r$}
            \begin{array}{l}
                \text{For every typeset $\dc\in\TSalt$ the query
                    $\query_\dc\deff\replace{\query_{r-1}}{x_r}{\dc}$ does not satisfy}\\
                \text{$\Supp(\query_\dc,\Sample) \ge \supp$.}
            \end{array}
        \end{equation}
        and
        \begin{equation}\label{eq:algo-sim-notype}\tag*{$(**)_r$}
            \begin{array}{l}
                \text{For every variable $y\in\AV_{r-1}$ the query
                    $\query_y\deff\replace{\query_{r-1}}{x_r}{y}$ does not satisfy}
                \\
                \text{$\Supp(\query_y,\Sample) \ge \supp$.}
            \end{array}
        \end{equation} 
        \begin{Claim}\label{claim:algo-noChangeOp}
            Let $p_r$ be the position of $x_r$ in $s_0$.
            \begin{itemize}
                \item
                \ref{eq:algo-sim-vars} implies that $s'[p_r]\in\VS$.
                \item
                \ref{eq:algo-sim-notype} implies that
                $\InPotTS{s'[p_r]}$ or
                $s'[p_r]\neq s'[j]$ for all $j\in I_{r-1}$.
            \end{itemize}
            \vspace*{-1em}
        \end{Claim}

        \begin{proof}
            Let us first focus on the claim's first statement. Let $p_r$ be the 
            position of $x_r$ in $s_0$ and let \ref{eq:algo-sim-vars} be 
            satisfied, i.e. there exists no $\dc\in\TSalt$ such that replacing 
            $x_r$ by $\dc$ yields a query that satisfies $\supp$.
            For contradiction, assume
                $s'[p_r] = \dcInPotTS$.
            By the choice of $\TSalt$ and since 
            $\Supp(\query',\Sample)\geq\supp$, we know that $\dc\in\TSalt$.
            For 
                $\query_\dc \deff \replace{\query_{r-1}}{x_r}{\dc}$
            we have 
                $\query_\dc \Hom \query'$
            because $\query_{r-1} \Hom \query'$ and $s'[p_r] = \dc$.
            Therefore, 
                $\Supp(\query_\dc,\Sample)\geq \Supp(\query',\Sample)\geq\supp$,
            contradicting \ref{eq:algo-sim-vars}. This completes the proof of 
            the first statement.      
            \par 
            Let us now turn to the second statement of 
            Claim~\ref{claim:algo-noChangeOp}. Let $p_r$ be the 
            position of $x_r$ in $s_0$ and let \ref{eq:algo-sim-notype} be 
            satisfied, i.e. there exists no available variable $y\in\AV_{r-1}$
            such that replacing $x_r$ by $y$ yields a query that satisfies 
            $\supp$.
            If $\InPotTS{s'[p_r]}$ 
            we are done.
            Consider the case where $s'[p_r]\in\Vars$ and assume for 
            contradiction that there exists a $k\in I_{r-1}$ such that  
                $s'[p_r] = s'[k]$.
            Let $y\deff s_{r-1}[k]$. From $\query_{r-1}\Hom\query'$ and 
            $s'[k]\in\Vars$ we obtain $y\in\Vars$. Since $k\in I_{r-1}$ we then
            obtain that there is a 
                $\nu\in\{1,\ldots, r-1\}$ 
            such that $\{k\}\in\Pos{s_0}{x_\nu}$. 
            We claim that $x_\nu\in\AV_{r-1}$.
            For contradiction, assume that $x_\nu\not\in\AV_{r-1}$. By 
            Claim~\ref{claim:algo-invariants}\eqref{item:algo-invariants:NRandAV}
            we have 
                $\vars(s_{r-1}) = \NR_{r-1} \cup \AV_{r-1}$,
            and hence $y\in\NR_{r-1}$.
            From Claim~\ref{claim:algo-invariants}\eqref{item:algo-invariants:NR-nonrepeated}
            we obtain that 
                $\Pos{s_{r-1}}{y} = \Pos{s_{0}}{y}$.
            Hence, 
                $\{k\} = \Pos{s_{0}}{y} \cap \Pos{s_{0}}{x_\nu}$.
            This implies that 
                $y=x_\nu$ 
            and due to $x_\nu\in\{x_1,\ldots,x_{r-1}\}$ it holds that 
                $y\in\{x_1,\ldots,x_{r-1}\}$. 
            But this is a contradiction to 
                $y\in\NR_{r-1}=\vars(\query_0)\setminus\{x_1,\ldots,x_{r-1}\}$.
            Thus, we have shown that $y\in\AV_{r-1}$.
            \par 
            Consider the query 
                $\query_y\deff\replace{\query_{r-1}}{x_r}{y}$,
            and let $s_y$ be the query string of $\query_y$. It holds that 
                $\query_y \Hom \query'$
            since we already know that 
                $\query_{r-1} \Hom \query'$.
            I.e., there is a homomorphism 
                $h: (\Vars \cup \Pfinplus(\TS))
                    \to 
                    (\Vars\cup \Pfinplus(\TS))$
            from $\query_{r-1}$ to $\query'$. This $h$ also is a homomorphism
            from $\query_y$ to $\query'$. To see this, note that for 
                $\{p_r\}\in\Pos{s_{r-1}}{x_r}$
            we have 
                $h(s_y[p_r]) = h(y) = h(s_{r-1}[k]) = s'[k] = s'[p_r]$;
            and for every other position $j\in[\ell]\setminus\Pos{s_{r-1}}{x_r}$
            we have
                $h(s_y[j]) = h(s_{r-1}[j]) = s'[j]$.
            Therefore, $\Supp(\query_y,\Sample)\geq\supp$ which contradits 
            \ref{eq:algo-sim-notype}.
            This ends the proof of Claim~\ref{claim:algo-noChangeOp}.\let\qed\relax
        \end{proof}
        To complete the proof of Claim~\ref{claim:SimInduction} we now 
        distinguish between the three cases depending on whether the query 
        $\query_r$ is obtained from $\query_{r-1}$ by perfomring a \TypeRepl,
        a \VarRepl, or a \NoChange. Our aim to show that in all cases the 
        conditions \eqref{item:sim-todo-i} and \eqref{item:sim-todo-ii} are 
        satisfied. Let us briefly recall these conditions:
        \begin{enumerate}[i]
			\item For all $j\in I_{r-1}$:
            \ If $s'[j]=s'[p_r]$ then $s_r[j]=s_r[p_r]$.\\
			\item
            If $\InPotTS{s'[p_r]}$ then 
            $s_r[p_r]\subseteq s'[p_r]$.
		\end{enumerate}
        \emph{Case 1:} $\query_r$ is obtained from $\query_{r-1}$ by replacing 
        the current variable $x_r$ in $s_{r-1}$ by a typeset $\dc\in\TSalt$, 
        i.e. $s_r = \replace{s_{r-1}}{x_r}{\dc}$. 
        Let $p_r$ the position of $x_r$ in $s_{r-1}$.
        By \eqref{item:sim-known-ii} we have
            $\varnothing\subset s'[p_r] \subseteq s[p_r] = \dc$.
        Since $\TSalt$ is walked through incrementally in 
        line~\ref{algo:forDeltaiLoopLine} $\dc$ is minimal in the following 
        sense: for each $\dc'\in\TSalt$ with $\dc'\subset\dc$ it holds that 
            $\Supp(\replace{\query_{r-1}}{x_r}{\dc'})<\supp$.
        Hence,
            $s[p_r] = \dc \subseteq s'[p_r]$
        holds as well, since otherwise $s'[p_r] = \dc'\subset\dc$ with 
        contradicts $\Supp(\query_r,\Sample)\geq \supp$.
        In particular \eqref{item:sim-todo-ii} is satisfied.
        To see that \eqref{item:sim-todo-i} is satisfied, consider 
            $\{p_r\}\in\Pos{s_0}{x_r}$
        and a 
            $j\in I_{r-1}$
        with 
            $s'[p_r] = s'[j] = \dc$.
        From $j\in I_{r-1}$ and $\query_r \sim_{I_{r-1}} \query'$ we obtain that 
            $s_r[p_r] = \dc = s_r[j]$.
        Hence, \eqref{item:sim-todo-i} is satisfied.
        \par 
        \emph{Case 2:} $\query_r$ is obtained from $\query_{r-1}$ by replacing 
        the variable $x_r$ in $s_{r-1}$ by an available variable 
        $y\in\AV_{r-1}$, i.e. the query string of $\query_r$ is 
            $s_r = \replace{s_{r-1}}{x_r}{y}$.
        According to Claim~\ref{claim:algo-invariants}\eqref{item:algo-invariants:NRandAV}
        there exists an 
            $r'\leq r-1$
        such that 
            $y=x_{r'}$.
        Furthermore, by definition of the algorithm, a variable can only be 
        included into the set $\AV$ in case of a \NoChange, i.e. neither a 
        \TypeRepl{} nor a \VarRepl{} was possible.
        Therefore, in the $r'$-th iteration of the algorithm's main loop, i.e. 
        the outer while-loop, the variable $x_{r'}$ was included into the set
        $\AV$. But this means that the conditions \ref{eq:algo-sim-vars} and 
        \ref{eq:algo-sim-notype} are satisfied. 
        Let $\{p_{r'}\}\in\Pos{s_0}{x_{r'}}$. The first statement of 
        Claim~\ref{claim:algo-noChangeOp} tells us that 
            $s'[p_{r'}]\in\Vars$.
        Note that 
            $p_{r'} \in I_{r'} \subseteq I_{r''}$
        for all $r''\geq r'$. 
        Hence, by Claim~\ref{claim:algo-invariants}\eqref{item:algo-invariants:PositionsNotInNRremainUnchanged}
        we obtain that 
            $y = x_{r'} = s_{r'}[p_{r'}] = s_{r''}[p_{r'}]$
        for all $r''\geq r'$.
        In particular, for $r''=r$ we obtain that 
            $y = s_r[p_{r'}]$.
        Hence, we have 
            $s_r[p_{r'}] = y = s_r[p_r]$
        for $\{p_r\}\in\Pos{s_0}{x_r}$.
        From $\query_r\Hom\query'$ we obtain that 
            $s'[p_{r'}] = s'[p_r]$.
        Since $s'[p_{r'}]\in\Vars$ we obtain that $s'[p_r]\in\Vars$. This 
        proves that condition \eqref{item:sim-todo-ii} is satisfied.
        \par 
        Let us now turn to condition \eqref{item:sim-todo-i}. Let 
        $\{p_r\}\in\Pos{s_0}{x_r}$ and choose an arbitrary $j\in I_{r-1}$ such 
        that 
            $s'[j] = s'[p_r]$.
        We want to prove that $s_r[j] = s_r[p_r]$.
        As shown above, 
            $s'[j] = s'[p_r] = s'[p_{r'}]$.
        From $j,p_{r'}\in I_{r-1}$ and $\query_r\sim_{I_{r-1}}\query'$ we obtain
        that 
            $s_r[j] = s_r[p_{r'}]$. 
        And we already know that 
            $s_r[p_{r'}] = y = s_r[p_r]$.
        This proves condition \eqref{item:sim-todo-i}.

        \par 
        \emph{Case 3:} $\query_r$ is obtained from $\query_{r-1}$ by perfomring 
        no replacement operation at all. In this case we know that the 
        statements \ref{eq:algo-sim-vars} and \ref{eq:algo-sim-notype} are 
        satisfied. From Claim~\ref{claim:algo-noChangeOp} we obtain for 
            $\{p_r\} \in \Pos{s_0}{x_r}$
        that $s'[p_r]\in\Vars$ and 
            $s'[j] \neq s'[p_r]$
        for all $j\in I_{r-1}$.
        Hence, both \eqref{item:sim-todo-i} and \eqref{item:sim-todo-ii} are 
        trivially satisifed.
        \par 
        In all three cases we have shown that \eqref{item:sim-todo-i} and 
        \eqref{item:sim-todo-ii} are satisfied, and thus we have 
            $\query_r \sim_{I_r} \query'$.
        This completes the proof of Claim~\ref{claim:SimInduction}. \let\qed\relax
    \end{proof} 
    In summary the proof of Proposition~\ref{prop:ComputeDescQueryFromq-dswg} is
    now completed. 
\end{proof}

\end{document}